\documentclass[journal,twoside]{IEEEtran}
\usepackage{graphicx}
\usepackage [noadjust]{cite} %[nospace]
\usepackage{bm}  % bold math for greek symbols
\usepackage{amsmath}
\usepackage{cite}
\usepackage{amssymb}
\usepackage{enumerate}
\usepackage{stfloats}
\usepackage{cases}
\usepackage{epstopdf}
\usepackage{soul,color,amsthm} % \ul  \hl
\usepackage{hyperref}
\usepackage[percent]{overpic}
\usepackage{tabularx}
\usepackage[table]{xcolor}
\usepackage{multirow}
\usepackage{booktabs}  % multicolumn tabular
\usepackage{tabu} % different font sizes in different rows of tabular

\usepackage{algorithm}
\usepackage{algpseudocode}% 

\newtheorem {definition}{Definition}

\newtheorem{Remark}{Remark}

\newtheorem{Proposition}{Proposition}
\begin{document}
    \title{Multi-Objective Energy Efficient Resource Allocation and User Association for In-band Full Duplex Small-Cells}
\author{
		\IEEEauthorblockN{Sheyda Zarandi, Ata Khalili, \textit{Graduate Student Member, IEEE},~Mehdi Rasti, \textit{Member, IEEE},~and Hina Tabassum,~\textit{Senior Member, IEEE}} 
\thanks{This work was supported in part by the  Discovery
	Grant from the Natural Sciences and Engineering Research Council of Canada. Sheyda Zarandi was at the Department of Computer Engineering, Amirkabir University of Technology, Tehran, Iran and now is with the Department of Electrical Engineering and Computer Science at  York  University, Canada.~Ata Khalili was a Visiting Researcher at the Department of Computer Engineering, Amirkabir University of Technology, Tehran, Iran and now is Research Assistant at Electronics Research Institute, Sharif University of Technology,
Tehran, Iran. Mehdi Rasti is with the Department of Computer Engineering, Amirkabir University of Technology, Tehran, Iran. Hina Tabassum is with the Department of Electrical Engineering and Computer Science at  York  University, Canada.(e-mails:~shz@york.ca,~ata.khalili@ieee.org, rasti@aut.ac.ir, hina@eecs.yorku.ca).~This paper has been presented in part at the IEEE Globecom 2019~\cite{GL}.}.}
			
\maketitle

	\begin{abstract}
	In this paper, we develop a framework to maximize the network energy efficiency (EE) by optimizing joint user-base station~(BS) association,~subchannel assignment, and power control considering an in-band full-duplex (IBFD)-enabled small-cell network. We maximize EE (ratio of network aggregate throughput and power consumption) while guaranteeing a minimum data rate requirement in both the uplink and downlink. The considered problem belongs to the category of mixed-integer non-linear programming problem (MINLP), {\color{black} and thus is NP-hard}. To cope up with this complexity and to derive a trade-off between system throughput and energy utilization, we first restate the considered problem as a multi-objective optimization problem (MOOP) aiming at maximizing system's throughput and minimizing system's energy consumption, simultaneously. This MOOP is then tackled  by using $\epsilon$-constraint method. To do so, we first transform the binary subchannel  and BS assignment variables into continuous ones without altering the feasible region of the problem and then approximate the non-convex rate functions through majorization-minimization (MM) approach. Simulation results are presented to demonstrate the effectiveness of our proposed algorithm in improving network's EE compared to the existing literature.~Furthermore, simulation results unveil that by employing the IBFD capability in OFDMA networks, our proposed resource allocation algorithm achieves a $69\%$ improvement in the EE as compared to the half-duplex system for practical values of residual self-interference.
	\end{abstract}
	\vspace{-2mm}
	\section{Introduction}
In-band full-duplex (IBFD) communication is a well-known  technique to enhance the spectral and energy efficiency of  emerging fifth generation (5G) wireless networks, by allowing  a user to send and receive data simultaneously in the same time and frequency (albeit at the cost of additional self-interference (SI))\cite{GL},\cite{111}. Energy efficient radio resource management algorithms generally aim to maximize the system throughput and minimize the corresponding energy consumption, without differentiating between the priority of these competing objectives. However, under certain circumstances,  handling the precedence of multiple objectives over each other becomes crucial. For example,  it is  more beneficial for network devices enabled with renewable energy to exploit network's resources for improving their quality-of-service (QoS) rather than focusing on minimizing their energy utilization.{\color{black}~Accordingly,~finding a trade-off between spectral efficiency (SE) and energy efficiency (EE) in IBFD communication through multi-objective optimization is crucial, specially since IBFD is in fact a promising technology to improve both of these metrics in UL as well as DL of cellular wireless networks \cite{DR.Hina,Overview}.}
	
	\subsection{Related Works}
    Recently, a plethora of research works considered IBFD communication to improve systems' performance~\cite{8,9,10,11,13,NOMA_FD,2,6,moghayese}. For instance, the problem of system throughput maximization in IBFD networks was investigated in \cite{2,8,9,10,11,13}. In \cite{8,9,10}, joint subchannel and power allocation  with one full-duplex base station (BS)  and multiple half-duplex (HD) user equipment was considered. In \cite{8}, the aforementioned problem was addressed when full channel state information (CSI) is known as well as when BS obtains limited CSI through channel feedbacks. In \cite{9}, an iterative algorithm for joint subchannel and power allocation was proposed, which deals with the power allocation after obtaining subchannel allocation through  variable relaxation. To avoid the time-complexity of the iterative algorithms, authors of \cite{10} employed decomposition method and dealt with power control  after obtaining subchannel allocation policy by using a heuristic approach. Similarly, in \cite{11} and \cite{13}, system throughput maximization was investigated for two-tier heterogeneous networks. In \cite{11}, the problem of subchannel assignment, power control, and duplexing mode selection was addressed using heuristic algorithms, while in \cite{13}, only the problem of power allocation was investigated considering both SI and cross-tier interference.~The authors in \cite{NOMA_FD} considered throughput maximization via joint subchannel assignment and power control in non orthogonal multiple access-FD system and obtained both locally optimal and suboptimal solutions.

    %,  while in \cite{MO1,MO2,MO3}, a multi-objective framework has been proposed for resource allocation in cellular networks. 

	In \cite{2}, the problem of decoupled uplink (UL)-downlink (DL) user association in a two-tier full duplex cellular network was considered and a geometric programming approach was proposed to maximize network throughput.~Furthermore,~a distributed many-to-one matching game based solution was proposed.~However,~the subchannel allocation was completely overlooked.~In \cite{6}, two distributed power allocation algorithms, one for minimizing network aggregate power consumption and one for maximizing network's throughput, were proposed.~However, QoS was guaranteed in terms of users' minimum signal to interference plus noise ratio (SINR).
%power allocation, this time for minimizing aggregate power consumption of a network with one FD access node and HD users is investigated, while. Both \cite{2} and \cite{6}, take users' QoS requirement into account, however in the latter, QoS is guaranteed in terms of users' minimum signal to noise and interference ratio (SINR).
	In \cite{moghayese}, subchannel and power allocations were optimized to maximize EE of a single cell IBFD network, while considering SI. A heuristic approach for obtaining subchannel allocation was proposed and then the power allocation problem was addressed using Augmented Lagrangian method.
	
	None of the aforementioned resource management algorithms for IBFD networks considered multi-objective optimization. Recently, IBFD communications was considered in a single cell  simultaneous wireless information and power transfer (SWIPT) network with IBFD BS and HD users in \cite{MO2}. The goal of the modeled multi-objective optimization problem (MOOP) was to derive a trade-off between minimizing UL and DL transmit powers as well as maximizing harvested energy. This problem was then optimally solved using semi-definite program relaxation. In \cite{MO3}, a framework for deriving a trade-off between EE and SE was proposed for a system with  IBFD small base stations (SBS) and HD users. The length of the time-slot allocated to each user-pair and users' UL and DL data rates were considered as variables and it is guaranteed that users' assigned time-slot is longer than a minimum value. Furthermore, in \cite{GL}, the problem of joint subchannel assignment and power control is studied to strike a balance between EE and SE in a single cell IBFD network.
	%In \cite{MO1}, resource allocation for obtaining the trade-off between EE and SE in a  a single-link network was addressed. Since in absence of interference the rate function is convex,  the trade-off between the aforementioned objective functions is achieved through a simple algorithm, using the weighted sum method. 
	\subsection{Contributions}
	To the best of our knowledge, the problem of joint BS,~subchannel assignment and power allocation for EE maximization in a  small-cell IBFD network under QoS and power feasibility constraints has not been investigated so far. Most of the aforementioned literature focus either on the throughput maximization  \cite{2,8,9,10,11,13,NOMA_FD}  with no  energy efficiency considerations or minimization of system's aggregate power consumption \cite{6, MO2} with no considerations to achievable throughput. Furthermore, except for \cite{11} and \cite{13}, all  aforementioned works consider single cell networks and, subsequently, the developed algorithms may not be directly applicable to large-scale networks with co-channel interference. Moreover, in \cite{8,9,10,11}, not only users' QoS requirements is completely ignored, but also  subchannel and power allocation problems are addressed separately, which results in performance degradation. 

	Considering the above literature review, our contributions can be summarized as follows:
	{\color{black}
	    \begin{itemize}
		\item To the best of our knowledge, joint optimization of user association (or BS assignment), power and subchannel assignment in the presence of IBFD communications has not been considered in the previous literature. In [5]-[10] and [12], only subchannel and power allocations are considered and are based on the alternating optimization where the variables are decomposed and optimized separately in an iterative manner. Since transmit power and subchannel allocation are closely intertwined variables, joint optimization of these resources can considerably enhance the performance of the system. In contrast to [5]-[10], and [12], in this paper, we  exploit the  benefits of joint resource allocation and investigate a joint optimization scheme for BS and subchannel assignment and power allocation for EE maximization.
        	
		\item BS assignment in IBFD communications requires considering both the UL and DL channel conditions, the severe SI, and inter-cell interference, which are all tightly dependant on the other optimization variables (transmit power and subchannel allocation variables). Due to these complexities, the problem of BS assignment for EE maximization in IBFD networks has not been properly addressed in previous literature. For instance in [11], which is one of the very few papers that focus on user association in IBFD networks, not only the objective function is maximization of system throughput, but also the UL and DL connections are completely decoupled from one another. Since the major benefit of IBFD communications is the capability of considering both UL and DL simultaneously, we consider coupled UL and DL IBFD communication without using any alternating optimization techniques.
		
		\item  To solve the EE maximization problem, we propose a multi-objective optimization (MOOP) framework. By doing so, we exploit numerous benefits of multi-objective optimization. Also,  we obtain a Pareto front that we will prove to contain the locally optimal solution of the original non-convex problem. We present the computational complexity  analysis  of  the  proposed  algorithm and demonstrate that this complexity is much lower for our algorithm compared to the conventional algorithms.
		
		\item Simulation results illustrate that our proposed algorithm outperforms the state-of-the-art algorithms in terms of EE. Also, by employing the IBFD in OFDMA networks with efficient SI cancellation, our proposed  algorithm achieves a $69\%$ improvement in the EE as compared to the half-duplex system for practical values of residual SI.
	\end{itemize}
		   }
	The rest of this paper is organized as follows. The system model is introduced in Section~II. In Section~III, problem formulation and the proposed solution is presented. The time complexity and performance of the proposed algorithm are evaluated in Section~IV and Section~V, respectively. Finally, Section~VI concludes this paper.
	
		\begin{figure*}
		\centering
	\includegraphics[width=18cm,height=8cm]{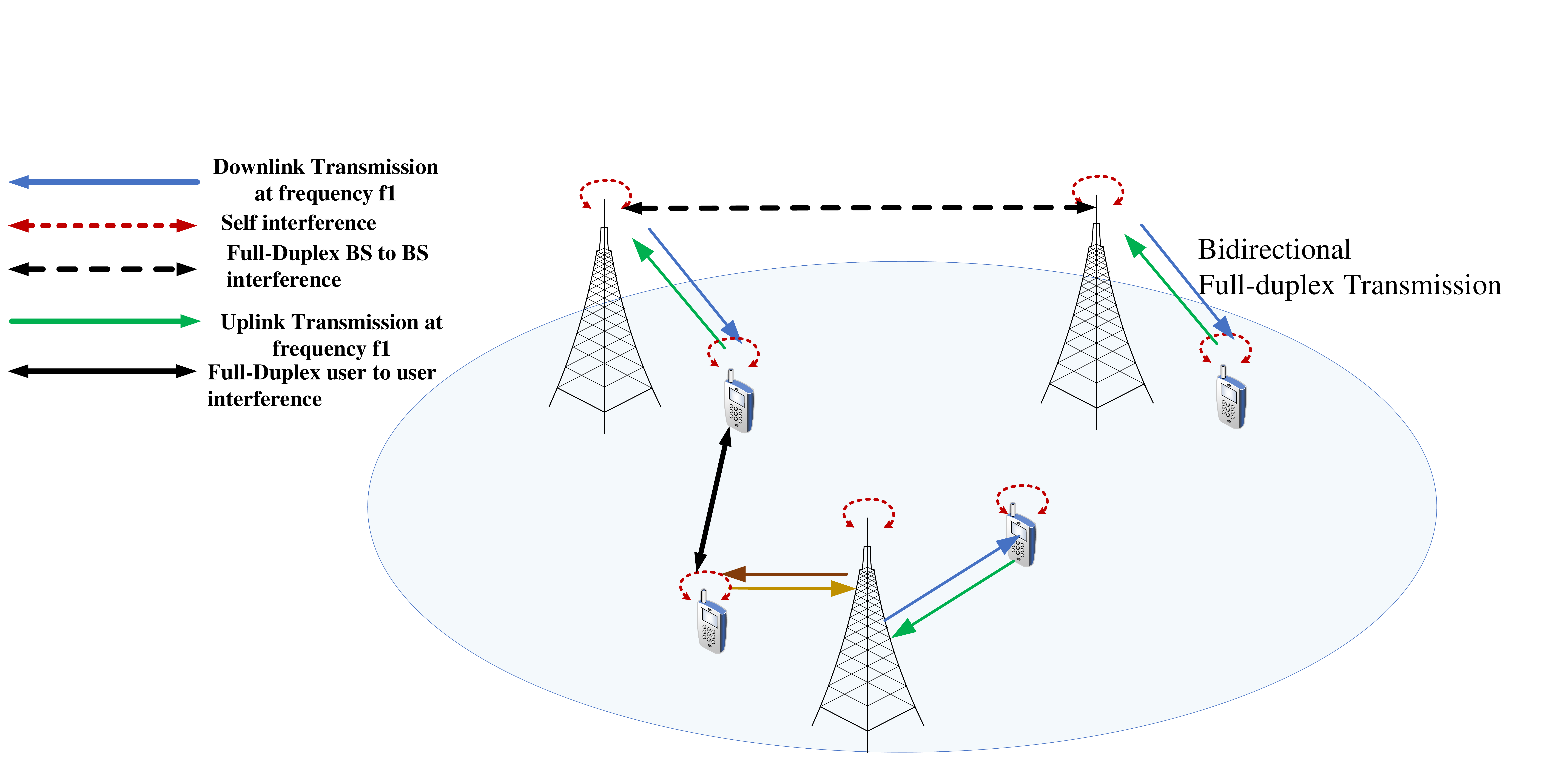}
	%\vspace{-10mm}
		\caption{\textcolor{black}{The considered network model of small-cells enabled with IBFD transmission.}}
	\end{figure*}
	
	\section{System Model and Assumptions}
In this paper we consider OFDMA small-cell network with $B$ SBSs and $N$ users, which are all capable of performing IBFD communications.~We assume that the entire frequency band is partitioned into $K$ subchannels each with bandwidth $\omega$. Furthermore, the set of SBSs, users, and subchannels are denoted by $\mathcal{B}=\{1,2,...,B\}$,~ $\mathcal{N}=\{1,2,...,N\}$,~ and $\mathcal{K}=\{1,2,...,K\}$, respectively. It is also assumed that all the subchannels are perfectly orthogonal to one another and no inter-subchannel interference exists. We also consider that a subchannel is exclusively assigned for the communications of a single user in both UL and DL in each cell. The joint BS and ~subchannel assignment variable is denoted by $a_{n,b,k}$ where
	\vspace{-1mm}
	\begin{equation}
	a_{n,b,k} = \begin{cases}
	1, & \text{If BS b allocates subchannel $k$ to user \textit{n},} \\
	0, & \text{otherwise}.
	\end{cases}
	\end{equation}
	
	Even though this exclusive subchannel allocation assumption may degrade the overall network's performance, due to the inter-user interference from an UL user over the user that uses the same subchannel for DL communications, it is highly unlikely for an optimal solution to assign a subchannel to two different nodes \cite{exlusive}. We further assume that $p_{n,b,k}$ and $q_{n,b,k}$ are the DL and UL transmit powers of user $n$ to SBS $b$, over subchannel $k$, respectively, and that $\mathbf{p}\in \mathbb{R}^{1 \times NBK}$ and $\mathbf{q}\in \mathbb{R}^{1 \times NBK}$ are vectors containing DL and UL transmit powers, in that order. Note that, $h_{s,s',k}$ denotes the channel coefficient between sender $s$ and receiver $s'$ over subchannel $k$. As such, the DL SINR of user $n$ in cell $b$ over subchannel $k$ is formulated as
	\begin{equation}
	\gamma_{n,b,k}^\mathrm{dl}(\mathbf{p},~\mathbf{q}) = \frac{p_{n,b,k} h_{b,n,k}} {\underbrace{{q_{n,b,k} \Delta_\mathrm{u}}}_{\text{residual SI in DL}}+ \underbrace{I_{n,b,k}^\mathrm{dl}}_{\text{inter-cell interference in DL}}+ \sigma^2},
	\end{equation}
	where $\sigma^2$ and $\Delta_\mathrm{u}$ denote the noise density and SI-cancellation factor of user devices, respectively. Since subchannel $k$ is used for communications of user $n$ in both directions, UL and DL signals will interfere with one another, which results in SI. In equation (2), the term that represents this residual SI in DL is clearly specified. Furthermore, $I_{n,b,k}^\mathrm{dl}$ represents the DL inter-cell interference that user $n$ in cell $b$ experiences over subchannel $k$ and is calculated as given in (3).
	\small
	\begin{equation}
	I_{n,b,k}^\mathrm{dl} = \sum_{b' \in {\mathcal{B}}, b' \neq b} \sum_{m \in {\mathcal{N}}, m \neq n} a_{m,b',k} \big(p_{m,b',k} h_{b',n,k} + q_{m,b',k} h_{m,n,k}\big).
	\end{equation}
	\normalsize
	Similarly, we define the UL SINR of user $n$ in cell $b$ over subchannel $k$ as
	\begin{equation}
	\gamma_{n,b,k}^\mathrm{ul}(\mathbf{p},~\mathbf{q}) = \frac{q_{n,b,k} h_{n,b,k}} {p_{n,b,k} \Delta_\mathrm{bs}+ I_{n,b,k}^\mathrm{ul}+ \sigma^2},
	\end{equation}
	with $\Delta_\mathrm{bs}$ and $I_{n,b,k}^\mathrm{ul}$ denote SI-cancellation factor of SBSs and inter-cell interference of user $n$ in UL over subchannel $k$, respectively, where
	\begin{equation}
	I_{n,b,k}^\mathrm{ul} = \sum_{b' \in {\mathcal{B}}, b' \neq b} \sum_{m \in {\mathcal{N}}} a_{m,b',k} \big(p_{m,b',k} h_{b',b,k} + q_{m,b',k} h_{m,b,k}\big).
	\end{equation}
	
	Note that SI is not the only additional type of interference that networks' nodes have to deal with in IBFD networks. In fact, as the result of increased frequency reuse factor in IBFD communications, inter-cell interference is also significantly intensified. For instance, as can be seen in (3) and (5), in IBFD communications, users as well as SBSs can cause interference on one another, whereas in traditional HD systems, due to the separation of UL and DL frequency bands, such interference did not exist.
	
	Based on the Shannon formula, the data rate of user $n$ in cell $b$ over subchannel $k$ in DL and UL is given, respectively, as follows:
	\begin{equation}\label{rate}
	r_{n,b,k}^\mathrm{dl} (\mathbf{a},~\mathbf{p},~\mathbf{q}) = a_{n,b,k}~\log_2{(1+\gamma_{n,b,k}^\mathrm{dl}(\mathbf{p},~\mathbf{q}))},
	\end{equation}
	and
	\begin{equation}\label{rate}
	r_{n,b,k}^\mathrm{ul} (\mathbf{a},~\mathbf{p},~\mathbf{q}) = a_{n,b,k}~\log_2{(1+\gamma_{n,b,k}^\mathrm{ul}(\mathbf{p},~\mathbf{q}))},
	\end{equation}
	 where $\mathbf{a}\in \mathbb{Z}^{1 \times NBK}$ denotes the joint BS,~subchannel assignment vector. Accordingly, the total data rate of user $n$ in DL can be formulated as follows:
	\begin{equation}\label{totalrateuser}
	R_n^\mathrm{dl} (\mathbf{a},~\mathbf{p},~\mathbf{q}) =\sum_{b\in\mathcal{B}}\sum_{k\in \mathcal{K}}{r_{n,b,k}^\mathrm{dl}(\mathbf{a},~\mathbf{p},~\mathbf{q})}.
	\end{equation}
	Similar to (8), the total data rate of user $n$ in UL, denoted by $R_n^\mathrm{ul}(\mathbf{a},~\mathbf{p},~\mathbf{q})$, is obtained as $
	R_n^\mathrm{ul} (\mathbf{a},~\mathbf{p},~\mathbf{q}) =\sum_{b\in\mathcal{B}}\sum_{k\in \mathcal{K}}{r_{n,b,k}^\mathrm{ul}(\mathbf{a},~\mathbf{p},~\mathbf{q})}.
	$To guarantee users' QoS, a minimum data rate represented by ${R}_\mathrm{min}^\mathrm{dl}$ in DL and ${R}_\mathrm{min}^\mathrm{ul}$ in UL, should be provided for each user and we have
	\begin{equation}
	\begin{aligned}
	R_{n}^\mathrm{dl}(\mathbf{a},~\mathbf{p},~\mathbf{q})  \geq{R}^\mathrm{dl}_\mathrm{min},~~ \forall n\in \mathcal{N},\\
	R_{n}^\mathrm{ul}(\mathbf{a},~\mathbf{p},~\mathbf{q})  \geq{R}^\mathrm{ul}_\mathrm{min},~~ \forall n\in \mathcal{N}.
	\end{aligned}
	\end{equation}
	The total network throughput, denoted by $R^T(\mathbf{a},~\mathbf{p},~\mathbf{q})$, is obtained as follows:
	\begin{equation}\label{R^T}
	R^T(\mathbf{a},~\mathbf{p},~\mathbf{q}) = \sum_{n \in \mathcal{N}} \big(R_{n}^\mathrm{ul}(\mathbf{a},~\mathbf{p},~\mathbf{q})+R_{n}^\mathrm{dl}(\mathbf{a},~\mathbf{p},~\mathbf{q})\big),
	\end{equation}
	To compute the total energy consumption of network, we use the following energy consumption model in which both transmit power consumption and circuit energy consumption of devices are taken into account, and there are coefficients that represent the efficiency of power amplifiers in network devices, i.e.,
	\begin{equation}\label{energy}
	\begin{aligned}
	&E^T(\mathbf{a},~\mathbf{p},~\mathbf{q}) = \\&\sum_{b \in \mathcal{B}}\sum_{n \in \mathcal{N}}\sum_{k\in \mathcal{K}} a_{n,b,k}(\frac{1}{\kappa}p_{n,b,k}+\frac{1}{\psi}q_{n,b,k})
	+ N P^\mathrm{u}_\mathrm{c}+BP_c^\mathrm{bs}.
	\end{aligned}
	\end{equation}
	In the above equation, $P^\mathrm{u}_\mathrm{c}$ and $P^\mathrm{bs}_\mathrm{c}$ denote the circuit energy consumption of user device and SBS, respectively, and $\kappa$ and $\psi$ are power amplifier efficiency in SBS and user device, in that order.
	%		In this paper, we define SE as the total system throughput per unit of bandwidth which is stated as:
	%	\begin{equation}
	%	\eta_{SE}(\textbf{A},\textbf{P},\textbf{Q}) = \sum_{n \in \mathcal{N}} \big(R_{n}^{ul}(\textbf{A},\textbf{P},\textbf{Q})+R_{n}^{dl}(\textbf{A},\textbf{P},\textbf{Q})\big).
	%	\end{equation}
	In this paper, we define EE as the ratio of system throughput to the corresponding network energy consumption, and denote it by $\eta(\mathbf{a},~\mathbf{p},~\mathbf{q})$, where
	\begin{equation}
	\eta(\mathbf{a},~\mathbf{p},~\mathbf{q}) = \frac{R^T(\mathbf{a},~\mathbf{p},~\mathbf{q})}{E^T(\mathbf{a},~\mathbf{p},~\mathbf{q})}.
	\end{equation}
	
	%	In the following section, we will formulate the MOOP for joint subchannel and power allocation optimization for simultaneously maximizing $R^T (\textbf{A},\textbf{P},\textbf{Q})$ and minimizing $E^T (\textbf{A},\textbf{P},\textbf{Q})$ under QoS and transmission power  feasibility constraints.
	\section{Problem statement and Proposed Solution}
	In this section, we will first mathematically  model the optimization problem of joint BS,~subchannel assignment, and power allocation for EE maximization, and state the main challenges in tackling this problem. Then, we  present our proposed algorithm to tackle the stated problem.
	\subsection{Problem Formulation}
	The problem of joint BS,~subchannel assignment and power allocation for maximizing system's EE under QoS and maximum transmit power constraints, is formally stated as below:
	\begin{equation}\label{first}
	\begin{aligned}
	&\max_{\mathbf{a},~\mathbf{p},~\mathbf{q}} \eta(\mathbf{a},~\mathbf{p},~\mathbf{q})\\
	&\text{subject to:}\\
	&C_1: \sum_{n \in \mathcal{N}}\sum_{k\in \mathcal{K}} a_{n,b,k} p_{n,b,k}\leq p_\mathrm{max},~\forall b\in \mathcal{B},\\
	&C_2: \sum_{b \in \mathcal{B}}\sum_{k\in \mathcal{K}} a_{n,b,k} q_{n,b,k} \leq p^\mathrm{u}_\mathrm{max},~\forall n\in \mathcal{N},\\
	&C_3: R_n^\mathrm{dl}(\mathbf{a},~\mathbf{p},~\mathbf{q}) \geq R_\mathrm{min}^\mathrm{dl},~ \forall n\in \mathcal{N},\\
	&C_4: R_n^\mathrm{ul}(\mathbf{a},~\mathbf{p},~\mathbf{q}) \geq R_\mathrm{min}^\mathrm{ul},~ \forall n\in \mathcal{N},\\
	&C_5: \sum_{n \in \mathcal{N}} a_{n,b,k}\leq 1,~ \forall b\in \mathcal{B},~k\in \mathcal{K},\\
	&C_6: a_{n,b,k}\in \{0,1\},~ \forall b\in \mathcal{B},~ n\in\mathcal{N},~ \forall k \in \mathcal{K}.\\
	&\textcolor{black}{C_7: a_{n,b,k}+\hspace{-4mm}\sum_{b' \in \mathcal{B}, b \neq b'}\hspace{-4mm} a_{n,b',k'}\leq 1,~\forall n \in \mathcal{N},~b\neq b',~\forall k,k' \in \mathcal{K}.}\\
	\end{aligned}
	\end{equation}
	
	In the optimization problem (\ref{first}), constraints $C_1$ and $C_2$ are related to  transmit power feasibility. Constraint $C_1$ indicates that the total transmit power of SBSs should not exceed their maximum threshold which is denoted by $p_\mathrm{max}$, and $C_2$ restricts users maximum transmit power to $p_\mathrm{max}^\mathrm{u}$. {\color{black} It is worth mentioning that since we are considering IBFD communications, we have to take both user devices' and BSs' maximum transmit power thresholds into account}. In constraints $C_3$ and $C_4$, a minimum rate requirement is guaranteed for each user in DL and UL, respectively. Constraint $C_5$ indicates that each subchannel can be allocated to at most one user in each cell and in $C_6$, the binary nature of subchannel allocation variable is implied.~\textcolor{black}{Finally,~$C_7$ ensures that a given subchannel can be assigned to a given user by only one BS. That is, each user in the network is restricted to connect to only one BS through optimization process.}
	
	Due to the binary joint BS and subchannel allocation variables, the interference included in rate function, and the fractional form of the objective function, problem (\ref{first}) is a mixed-integer non-linear programming (MINLP) problem which is generally difficult to solve. The challenges that make the above optimization problem complicated are listed below.
	\begin{itemize}
		\item Objective function $\eta(\mathbf{a},~\mathbf{p},~\mathbf{q})$ is in the  fractional form  with respect to $\mathbf{p}$ and thus non-convex.
		\item Since the multiplication of two variables is non-convex, in constraint $C_1$, the term $a_{n,b,k}p_{n,b,k}$, and in $C_2$, $a_{n,b,k}q_{n,b,k}$ pose a challenge in tackling (\ref{first}). Moreover, since in both $r_{n,b,k}^\mathrm{dl} (\mathbf{a},~\mathbf{p},~\mathbf{q})$ and $r_{n,b,k}^\mathrm{ul} (\mathbf{a},~\mathbf{p},~\mathbf{q})$, joint BS and subchannel assignment variable is multiplied by a function of transmit power (as given in (6) and (7)), constraints $C_3$ and $C_4$ are also non-convex.
		\item The SI as well as the inter-cell interference incorporated in the rate function,~make both constraints $C_3$ and $C_4$ non-convex. 	
		\item Presence of binary joint BS and~subchannel assignment variable, $\mathbf{a}$, which turns (\ref{first}) into MINLP.
	\end{itemize}
	In the following section, we first restate problem (\ref{first}) as a MOOP, whose purpose is to maximize system throughput and minimize energy consumption, simultaneously. Afterwards, our proposed algorithm would be presented.
	%\vspace{-3mm}
	\subsection{Equivalent MOOP formulation of the Objective Function}
	As given in (12), EE is the ratio of system's aggregate throughput and energy consumption. Since both $R^T(\mathbf{a},~\mathbf{p},~\mathbf{q})$ and $E^T(\mathbf{a},~\mathbf{p},~\mathbf{q})$ are positive functions of $\mathbf{a}$, $\mathbf{p}$ and $\mathbf{q}$, we can conclude that maximization of $\eta(\mathbf{a},~\mathbf{p},~\mathbf{q})$ is equivalent to maximizing $R^T(\mathbf{a},~\mathbf{p},~\mathbf{q})$ while minimizing $E^T(\mathbf{a},~\mathbf{p},~\mathbf{q})$, simultaneously. To this end, we reformulate (\ref{first}) as an equivalent MOOP that is given in (\ref{second}).
	\begin{equation}\label{second}
	\begin{aligned}
	&f_1: \min_{\mathbf{a},~\mathbf{p},~\mathbf{q}}~E^{T}(\mathbf{a},~\mathbf{p},~\mathbf{q})\\
	&f_2: \max_{\mathbf{a},~\mathbf{p},~\mathbf{q}}~R^{T}(\mathbf{a},~\mathbf{p},~\mathbf{q})\\
	&\text{subject to:} ~~C_1 - C_7.
	\end{aligned}
	\end{equation}
	The first objective of the optimization problem (\ref{second}), $f_1$, is to minimize system energy consumption, and the second one, $f_2$, is to maximize system's throughput and its constraint set is the same as that of (\ref{first}). 
	\begin{Proposition}
	\textcolor{black}{The solution to the MOOP given in (14) includes the solution of the EE-maximization
problem introduced in (13)  as a special case.}
	\end{Proposition}
	\begin{proof}
	\textcolor{black}{Please see Appendix A.}
	\end{proof}
	{\color{black} Even though the MOOP (\ref{second}) contains two competing objective functions, we can still find a solution for it that satisfies the predefined conditions of Pareto optimality.~In contrast to the Dinkelbach method, which is only applicable to an optimization problem with fractional function, MOOP  can be adopted for any optimization problem even if the number of objective function exceeds two \cite{Ata_TWC,Mili}. Moreover,  MOOP provides a SE-EE trade-off with lower computational complexity compared to the Dinkelbach method.}
	\begin{definition}
		Assume that $\mathcal{X}$ denotes the feasible region of the optimization problem
		\begin{equation*}
		\min \{f_1(x),~f_2(x), ... , f_n(x)\} ~~ \text{subject to: } x \in \mathcal{X}.
		\end{equation*}
		
		A point $x^*\in \mathcal{X}$, is Pareto optimal, if and only if for any other point $x \in \mathcal{X}$, we  have  $f_i(x) \geq f_i(x^*)~\forall i \in \{1,2,...,n\}$, with at least one $f_i(x)> f_i(x^*)~\forall i \in \{1,2,...,n\}$.
	\end{definition}
	In order to obtain the Pareto optimal fronts for (\ref{second}), we employ $\epsilon$-constraint method, in which one of the objective functions is chosen as the primary objective and the rest of objective functions are moved to the constraint set \cite{article}. Since system throughput, $R^T(\mathbf{a},~\mathbf{p},~\mathbf{q})$, is by itself a function of system's transmit power, it can be said that the effect of energy consumption on system's EE is generally much more substantial than that of system throughput. Thus, in this paper we keep $f_1$ as the primary objective function and move $f_2$ to the constraint set. Using the $\epsilon$-constraint method, the new optimization problem would be:
	\begin{equation}\label{third}
	\begin{aligned}
	&f_1: \min_{\mathbf{a},~\mathbf{p},~\mathbf{q}}E^{T}(\mathbf{a},~\mathbf{p},~\mathbf{q})\\
	&\text{subject to:} \\
	& C_0: R^T(\mathbf{a},~\mathbf{p},~\mathbf{q}) \geq \epsilon,\\
	&C_1 - C_7.
	\end{aligned}
	\end{equation}
	\begin{Remark}
		Since $\epsilon$-constraint method generates the whole Pareto fronts \cite{article}, the solution set obtained by solving (\ref{third}) contains the solution of EE-maximization
		problem (\ref{first}).
	\end{Remark}
	%\begin{Proof}
	%See Appendix A
	%\end{Proof}
	%
	 Constraint $C_0$ in (15) requires the total throughput of the network to be greater than $\epsilon$. Due to the multiplication of variables, $\mathbf{a}$ with both  $\mathbf{p}$ and  $\mathbf{q}$, the objective function of (\ref{third}) is still non-convex and thus challenging to address.
	
	{\bf Significance of $\epsilon$:} It is obvious that the feasibility of (\ref{third}), as well as  the closeness of its solution to the solution of problem (\ref{first}), greatly depends on the value of $\epsilon$. This fact turns $\epsilon$ into a sensitive parameter, whose value plays a major role both in prioritizing the objective functions in (\ref{second}) and finding an energy efficient trade-off between them. Moreover, we are still faced with the same challenges in dealing with the non-convex constraint set of (\ref{first}).
	
	In the following subsection, we introduce our proposed algorithm  for dealing with the non-convexity of feasible set in (\ref{third}). Next, in section IV, we present our proposed method to estimate the value of $\epsilon$ that results in EE maximization.
	\vspace{-3mm}
	\subsection{Equivalent Reformulation of Constraints Involving Binary and Continuous Variable Product}
	In order to address the non-convex optimization problem (13), we first deal with the problem of variables multiplication in constraints $C_1$ and $C_2$. In the left-hand side of these constraints, it is implied that if subchannel $k$ is not allocated to user $n$ in cell $b$ ($a_{n,b,k} =0 $), the transmit power of this user over $k$ should be zero in both UL and DL ($p_{n,b,k}=q_{n,b,k}=0$). Based on this explanation, we can restate $C_1$ and $C_2$ as below:
	\begin{align}
	&~C'_1: \sum_{n \in \mathcal{N}}\sum_{k\in \mathcal{K}} p_{n,b,k}\leq p_\mathrm{max},~\forall b\in \mathcal{B},\\
	&~C''_1: p_{n,b,k} \leq a_{n,b,k}p_\mathrm{max},~\forall b\in \mathcal{B},~\forall n \in \mathcal{N},~\forall k \in \mathcal{K},\\
	&~C'_2: \sum_{b \in \mathcal{B}} \sum_{k\in \mathcal{K}} q_{n,b,k}\leq p^\mathrm{u}_\mathrm{max},~\forall n \in \mathcal{N},\\
	&~C''_2: q_{n,b,k} \leq a_{n,b,k}p^\mathrm{u}_\mathrm{max},~\forall b\in \mathcal{B},~\forall n \in \mathcal{N},~\forall k \in \mathcal{K}.
	\end{align}
	
	It is clear that the feasible region defined by (16) and (17) is equal to that of constraint $C_1$. Thus, we can substitute $C_1$ with $C'_1$ and $C''_1$ without altering the feasible set of (\ref{third}). Similarly, constraint $C_2$ is replaced by (18) and (19). Through this method, we can easily deal with the non-convex constraints $C_1$ and $C_2$, by using their equivalent convex forms \cite{CL}.
	
	Furthermore, since in (17) and (19), users are restricted to transmit only over their assigned subchannels (if $a_{n,b,k} =0$, $q_{n,b,k}$ and $p_{n,b,k}$ must be zero),  SINR of users and thus their data rate over subchannels that are not allocated to them would be zero. Therefore, we can restate (6) and (7) as $r_{n,b,k}^\mathrm{dl}(\mathbf{p},~\mathbf{q}) = \log_2(1+\gamma_{n,b,k}^\mathrm{dl})$ and $r_{n,b,k}^\mathrm{ul}(\mathbf{p},~\mathbf{q})= \log_2(1+\gamma_{n,b,k}^\mathrm{ul})$, respectively. Similarly, we can omit $a_{n,b,k}$ from (11) and restate the objective function of (\ref{third}), $E^T(\mathbf{a},\mathbf{p},\mathbf{q})$, as an affine function that is given by
	\small
	\begin{equation}
	\begin{aligned}
	E^T(\mathbf{p},~\mathbf{q}) = &\sum_{b \in \mathcal{B}}\sum_{n \in \mathcal{N}_b}\sum_{k\in \mathcal{K}} (\frac{1}{\kappa}p_{n,b,k}+\frac{1}{\psi}q_{n,b,k})
	+ N P^\mathrm{u}_\mathrm{c}+BP_c^\mathrm{bs}.
	\end{aligned}
	\end{equation}
	\normalsize
	\subsection{Equivalent Reformulation of Binary Constraints}
	Another challenge in solving (\ref{third}) is the integer joint BS and subchannel assignment variable, $a_{n,b,k}$. This binary variable turns (\ref{third}) into a MINLP, which is difficult to solve in an acceptable time span. To address this issue, we take an approach similar to \cite{Ata_WCL,subchannel}, and replace constraint $C_6$ with the following inequalities:
	\begin{align}
	&C'_6: 0 \leq a_{n,b,k} \leq 1,~\forall b\in \mathcal{B},~\forall n\in \mathcal{N}, \forall k \in \mathcal{K},\\
	&C''_6: \sum_{b\in \mathcal{B}.} \sum_{n\in \mathcal{N}}\sum_{k\in \mathcal{K}}(a_{n,b,k} -a^2_{n,b,k})\leq 0,
	\end{align}
	In $C'_6$ it is stated that the variable $a_{n,b,k}$ is continuous with values in the range [$0,1$]. However, in $C''_6$ the value of $a_{n,b,k}$ is restricted to 0 and 1, since the only two numbers that fit in (22), belong to the set $\{0,1\}$. Therefore, the intersection of constraints $C'_6$ and $C''_6$, is a region that is equivalent to that of $C_6$.
	Nevertheless, as $C''_6$ is concave and greater than or equal to zero, this constraint does not comply with the standard form of inequality constraints in convex optimization problems. 
	
	To deal with this issue, we remove constraint $C''_6$ from the constraint set and instead add it as a penalty function, with a weighting factor, denoted by $\lambda$, to the objective function. {In fact $\lambda$ acts as a penalty factor to penalize the objective function when $a_{n,b,k}$ is not binary.} After this modification, we get the equivalent reformulation of the original problem as follows:
	\begin{equation}\label{last+1}
	\begin{aligned}
	&\min_{\mathbf{a},~\mathbf{p},~\mathbf{q}} ~E^T(\mathbf{p},\mathbf{q})+\lambda\Big(\sum_{b \in \mathcal{B}}\sum_{n \in \mathcal{N}_b}\sum_{k\in \mathcal{K}}\big(a_{n,b,k}-a_{n,b,k}^2\big)\Big)\\
	& \text{subject to:}~~ C_0,~C'_1,~C'_2,~C''_1,~C''_2,~C_3,~C_4,~C_5,~C'_6,~C_7.
	\end{aligned}
	\end{equation}
	\begin{Remark}
		It can be easily demonstrated that the optimization problem (\ref{last+1}) is equivalent to (\ref{third}). For more details refer to \cite{CL}.
	\end{Remark}
	
	\subsection{Convex Approximation of the Objective Function via Majorization-Minimization}
	To tackle the non-convexity of objective function in the above problem, we first rewrite the objective function as follows:
	\begin{equation}
	e (\mathbf{a},~\mathbf{p},~\mathbf{q}) = e_1( \mathbf{a},~\mathbf{p},~\mathbf{q}) - \lambda e_2(\mathbf{a}),
	\end{equation}
	where $e_1( \mathbf{a},~\mathbf{p},~\mathbf{q})=E^T(\mathbf{p},~\mathbf{q})+\lambda \big(\sum_{b \in \mathcal{B}}\sum_{n \in \mathcal{N}}\sum_{k\in \mathcal{K}}a_{n,b,k}\big)$ and\\
	$e_2(\mathbf{a})=\big(\sum_{b \in \mathcal{B}}\sum_{n \in \mathcal{N}}\sum_{k\in \mathcal{K}}a_{n,b,k}^{2}\big)$.
	The equality given in (24) consists of two convex functions, $e_1(\mathbf{a},~\mathbf{p},~\mathbf{q})$ and $\lambda e_2(\mathbf{a})$. However, the subtraction of these convex functions is not necessarily convex. To tackle this issue, we find a convex approximation for $e(\mathbf{a},~\mathbf{p},~\mathbf{q})$ by using majorization minimization (MM) method \cite{MM}. In this method, a series of surrogate functions are constructed that approximate the originally non-convex function. Here, we use Taylor approximation for constructing our surrogate function. To do so, in iteration number $t$ we will have:
	\begin{equation}\label{four}
	\begin{aligned}
	\tilde{e_2}(\mathbf{\mathbf{a}})= e_2(\mathbf{a}^{t-1})+\nabla_{\mathbf{a}} e_2^T(\mathbf{a}^{t-1})(\mathbf{a}-\mathbf{a}^{t-1}).
	\end{aligned}
	\end{equation}
	
	Now we can replace $e_2(\textbf{a})$  with its affine approximation,
	$\tilde{e}_2(\textbf{a})$ and since subtraction of a convex function and an affine function is convex, the problem of non-convex objective function would be solved. Thus we will have
	\begin{equation}\label{last-1}
	\begin{aligned}
	&\min_{\mathbf{a},~\mathbf{p},~\mathbf{q}}e_1(\mathbf{a},~\mathbf{p},~\mathbf{q}) - \lambda \tilde{e}_2(\mathbf{a})\\
	& \text{subject to:}~~ C_0,~C'_1.~C'_2,~C''_1,~C''_2,~C_3,~C_4,~C_5,~C'_6,~C_7.
	\end{aligned}
	\end{equation}
	
	Even after the above transformations, due to the non-convexity of rate functions $R_n^\mathrm{ul}(\mathbf{p},~\mathbf{q})$ and $R_n^\mathrm{dl}(\mathbf{p},~\mathbf{q})$, optimization problem (\ref{last-1}) is still intractable. Let us rewrite $R_n^\mathrm{dl}(\mathbf{p},~\mathbf{q})$ as follows:
	\begin{equation}
	R_n^\mathrm{dl}(\mathbf{p},~\mathbf{q}) = f^\mathrm{dl}_{n}(\mathbf{p},~\mathbf{q}) - g^\mathrm{dl}_{n}(\mathbf{p},~\mathbf{q}),
	\end{equation}
	where
	\small
	\begin{equation}
	\begin{aligned}
	f^\mathrm{dl}_{n}(\mathbf {p},~\mathbf{q}) =\sum_{b\in\mathcal{B}}\sum_{k \in {\mathcal{K}}}&\log_2 (p_{n,b,k}h_{b,n,k}+q_{n,b,k}\Delta_\mathrm{u}+I_{n,b,k}^\mathrm{dl}+\sigma^2),
	\end{aligned}
	\end{equation}
	\normalsize
	and
	\begin{equation}
	g^\mathrm{dl}_{n}(\mathbf{p},~\mathbf{q}) = \sum_{b\in\mathcal{B}}\sum_{k \in {\mathcal{K}}} \log_2 (q_{n,b,k}\Delta_\mathrm{u}+I_{n,b,k}^\mathrm{dl}+\sigma^2).
	\end{equation}
	Now we can use MM approach and approximate $g_n^\mathrm{dl}(\textbf{p},~\textbf{q})$ as follows:
	\begin{align}
	\tilde{g}^\mathrm{dl}_{n}(\mathbf{p},~\mathbf{q}) = &g_{n}^\mathrm{dl}(\mathbf{q}^{t-1},\mathbf{p}^{t-1})+\nabla_\mathbf{q}g^T(\mathbf{q}^{t-1}). (\mathbf{q}-\mathbf{q}^{t-1})+\nonumber\\&\nabla_\mathbf{p}g^T(\mathbf{p}^{t-1}). (\mathbf{p}-\mathbf{p}^{t-1})
	\end{align}
	
	Thus the convex approximation of DL rate function, ${R}_n^\mathrm{dl}(\mathbf{p},~\mathbf{q})$, would be
	\begin{equation}
	\tilde{R}_n^\mathrm{dl}(\mathbf{p},~\mathbf{q}) = f^\mathrm{dl}_{n}(\mathbf{p},~\mathbf{q}) - \tilde{g}^\mathrm{dl}_{n}(\mathbf{p},~\mathbf{q}).
	\end{equation}
	
	Similarly, the approximate UL data rate is
	\begin{equation}
	\tilde{R}_n^\mathrm{ul}(\mathbf{p},~\mathbf{q}) = f^\mathrm{ul}_{n}(\mathbf{p},~\mathbf{q}) - \tilde{g}^\mathrm{ul}_{n}(\mathbf{p},~\mathbf{q}),
	\end{equation}
	where,
	\small
	\begin{equation}
	\begin{aligned}
	f^\mathrm{ul}_{n}(\mathbf{p},~\mathbf{q}) =\sum_{b\in\mathcal{B}}\sum_{k \in {\mathcal{K}}} &\log_2 (q_{n,b,k}h_{n,b,k}+p_{n,b,k}\Delta_\mathrm{bs}+I_{n,b,k}^\mathrm{ul}+\sigma^2),
	\end{aligned}
	\end{equation}
	\normalsize
	\begin{align}
	\tilde{g}^\mathrm{ul}_{n}(\mathbf{p},~\mathbf{q}) = &g_{n}^\mathrm{ul}(\mathbf{p}^{t-1},\mathbf{q}^{t-1})+\nabla_\mathbf{p}g^T(\mathbf{p}^{t-1}). (\mathbf{p}-\mathbf{p}^{t-1})+\nonumber\\&\nabla_\mathbf{q}g^T(\mathbf{q}^{t-1}). (\mathbf{q}-\mathbf{q}^{t-1})
	\end{align}
	and 
	\begin{equation}
	g^\mathrm{ul}_{n}(\mathbf{p},~\mathbf{q}) = \sum_{b\in\mathcal{B}}\sum_{k \in {\mathcal{K}}} \log_2 (p_{n,b,k}\Delta_\mathrm{bs}+I_{n,b,k}^\mathrm{ul}+\sigma^2).
	\end{equation}
	
	Regarding the above transformations, we define the approximate total data rate of system as:
	\begin{equation}
	\tilde{R}^{T}(\mathbf{p},~\mathbf{q}) = \sum_{b\in \mathcal{B}}\sum_{n \in \mathcal{N}} (\tilde{R}_n^\mathrm{dl}(\mathbf{p},~\mathbf{q})+ \tilde{R}_n^\mathrm{ul}(\mathbf{p},~\mathbf{q})).
	\end{equation}
	
	Finally, after these modifications, the resulting convex optimization problem would be:
	\begin{equation}\label{last}
	\begin{aligned}
	&\min_{\mathbf{a},~\mathbf{p},~\mathbf{q}}e_1(\mathbf{a},~\mathbf{p},~\mathbf{q}) - \lambda \tilde{e}_2(\mathbf{a})\\
	&\text{subject to:}\\
	&C_0: \tilde{R}^{T}(\mathbf{p},~\mathbf{q})\geq \epsilon,\\
	&C'_1: \sum_{n \in \mathcal{N}}\sum_{k\in \mathcal{K}} p_{n,b,k}\leq p_\mathrm{max},~\forall b\in \mathcal{B},\\
	&C'_2:\sum_{b \in \mathcal{B}} \sum_{k\in \mathcal{K}} q_{n,b,k} \leq p^\mathrm{u}_\mathrm{max},~\forall b\in \mathcal{B},~\forall n\in \mathcal{N},\\
	&C''_1: p_{n,b,k} \leq a_{n,b,k} p_\mathrm{max}, ~ \forall b\in \mathcal{B},~\forall n\in\mathcal{N},~\forall k \in \mathcal{K},\\
	&C''_2: q_{n,b,k}\leq a_{n,b,k} p^\mathrm{u}_\mathrm{max}, ~\forall b\in \mathcal{B},~ \forall n\in\mathcal{N},~\forall k \in \mathcal{K},\\
	&C_3: \tilde{R}_n^{dl}(\mathbf{p},~\mathbf{q}) \geq R_\mathrm{min}^\mathrm{dl},~ \forall n\in \mathcal{N},\\
	&C_4: \tilde{R}_n^{ul}(\mathbf{p},~\mathbf{q}) \geq R_\mathrm{min}^\mathrm{ul},~ \forall n\in \mathcal{N},\\
	&C_5: \sum_{n \in \mathcal{N}} a_{n,b,k}\leq 1,~ \forall b\in \mathcal{B},~\forall k\in \mathcal{K},\\
	& C'_6: 0\leq a_{n,b,k}\leq 1,~\forall b\in \mathcal{B},~\forall n\in\mathcal{N},\forall k \in \mathcal{K},\\
	&C_7: \sum_{b \in \mathcal{B}} a_{n,b,k}\leq 1,~ \forall n\in \mathcal{N},~k\in \mathcal{K}.\\
	\end{aligned}
	\end{equation}
	Optimization problem (\ref{last}) is a convex optimization problem. In order to solve this problem and obtain a locally optimal solution for problem (\ref{third}), here we employ the difference of convex functions (DC) programming \cite{DC}.~In DC programming, the iteration starts from a feasible initial point and iteratively solves the optimization problem and obtains a locally optimal solution eventually \cite{DC,NOMA_FD},~and [23].
	%It should be noted that in converting (\ref{third}) into (\ref{last}), we have used majorization minimization method to approximate data rate function and objective function. Since in [], it is proved that the result of this approximation is tight if DC programming is employed, we can conclude that the solution obtained by solving (\ref{last}), would be a locally optimal point for the optimization problem (\ref{third}).
	\begin{Proposition}
		The solution obtained for (\ref{last}) by incorporating DC approximation at the end of each iteration, is a locally optimal solution for the original problem (\ref{first}). Our proposed algorithm to solve (\ref{last}) is presented in \textbf{Algorithm 1}.
	\end{Proposition}
	\begin{proof}
	See Appendix B.
	\end{proof}

		\begin{algorithm}
		\caption{ Proposed Algorithm to solve Eq. (37)}
		\label{euclid}
		\begin{algorithmic}[1]
			\State Obtain the value of $R_\mathrm{max}$ by solving optimization problem (\ref{R_max}).
			\State Initialize iteration number $t=1$, $\delta = 0$, and step size $\nu$, with a positive value, and $\eta^{t} = 0$.
			\State \textbf{while} ($\delta \leq 1$)
			\State ~~~Set $\delta = \delta+\nu$.
			\State ~~~Calculate $\epsilon$ using (\ref{eps}).
			%		\State ~~~Solve (\ref{last}) and obtain $\mathbf{A}^t$,$\mathbf{P}^t$, $\mathbf{Q}^t$
			%\State ~~~Set
			\State~~~Initialize $i=0$, maximum number of iteration $I_{\max}$, penalty factor $\lambda\gg1$ , and $\textbf{p}^{0} $, $\textbf{q}^{0}$, and $\textbf{a}^{0}$.
			\State~~~\textbf{Repeat}
			\State~~~~Update $\tilde{e}_{2}(\textbf{a})$, $\tilde{g}_n^\mathrm{dl}(\textbf{p},~\textbf{q})$, and $\tilde{g}_n^\mathrm{ul}(\textbf{p},~\textbf{q})$ using (25), (30), and (34), respectively.
			\State~~~~Solve problem (\ref{last}) and obtain $\textbf{a}^{i}$,~${\textbf{p}^{i}}$, and $\textbf{q}^{i}$.
			\State~~~~Set $i=i+1$.
			\State~~~\textbf{Until} convergence or $i=I_{\max}$
			\State~~~Set $\textbf{a}^{*}=\textbf{a}^{i}$, $\textbf{p}^{*}=\textbf{p}^{i}$, $\textbf{q}^{*}=\textbf{q}^{i}$.
			\State ~~~$\eta^t(\mathbf{a}^*,\mathbf{p}^*,\mathbf{q}^*) = \frac{\tilde{R}^T(\mathbf{p}^*,\mathbf{q}^*)}{E^T(\mathbf{p}^*,\mathbf{q}^*)}$.
			\State\textbf{end}
		\end{algorithmic}
	\end{algorithm}
	
	\section{Choice of $\epsilon$ and Computational Complexity of Algorithm~1}
	As explained in the previous section, constraint $C_0$ in optimization problem (\ref{last}) asserts that the total throughput of network, $\tilde{R}^T(\mathbf{a},~\mathbf{p},~\mathbf{q})$, should be greater than or equal to $\epsilon$. To further clarify the impact of $\epsilon$ on the optimization problem (\ref{last}), let us consider the following cases:
	\begin{itemize}
		\item[i)]if $\epsilon = 0$, optimization problem (\ref{last}) minimizes network's energy consumption.
		\item[ii)] if $\epsilon = R_\mathrm{max}$, assuming $R_\mathrm{max}$ is the maximum system throughput, the solution obtained for (\ref{last}) would be the solution of network throughput maximization problem.
		\item[iii)]  if $\epsilon \geq R_\mathrm{max}$, the optimization problem (\ref{last}) would be infeasible.
		\item[iv)]  if $0 < \epsilon < R_\mathrm{max}$, the  problem (\ref{last}) would be a multi-objective optimization problem.
	\end{itemize}
	
	Regarding the above cases, it can be deduced that the optimization problem (\ref{last}) and its obtained solution is extremely sensitive to the value of $\epsilon$. Furthermore, any change in the priority of the objective functions can be achieved by manipulating the value of this parameter. Namely, when the chosen value for $\epsilon$ is high, more emphasis is put on system throughput maximization, while lower values of $\epsilon$ results in higher priority for system energy consumption minimization.
	
	Furthermore, according to the above three cases we can also conclude that for an specific value of $\epsilon$, a trade-off between system's throughput and aggregate energy consumption would be derived that results in maximum EE. To find this specific value of $\epsilon$, when our goal is to maximize EE, we proposed the following algorithm. 
	
	From cases (i) and (ii), we can perceive that the maximum value that $\epsilon$ can take without making (\ref{last}) infeasible is $R_\mathrm{max}$. Since $R_{\max}$ is maximum system throughput, we can obtain its value by solving the following optimization problem:
	\begin{equation}\label{R_max}
	\begin{aligned}
	&\max_{\mathbf{a},~\mathbf{p},~\mathbf{q}}\tilde{R}^{T}(\mathbf{p},~\mathbf{q})-\lambda\Big( \sum_{b \in \mathcal{B}}\sum_{n \in \mathcal{N}_b}\sum_{k\in \mathcal{K}}a_{n,b,k}-\tilde{e}_2(\mathbf{a})\Big)\\
	& \text{subject to:}~~ C_0,~C'_1,~C'_2,~C''_1,~C''_2,~C_3,~C_4,~C_5,~C'_6,~C_7.
	\end{aligned}
	\end{equation}
	which is in fact the optimization problem of maximizing system's throughput.
	By solving the optimization problem in (\ref{R_max}), the maximum value of $\epsilon$ to avoid infeasibility of the problem can be determined.

	\subsection{Choosing $\epsilon$ for Energy Efficiency Maximization}
	Since different values of $\epsilon$ results in different trade-offs between system's throughput and energy consumption, to maximize network's EE, we should find a value for $\epsilon$ that corresponds to the maximum $\tilde{R}^T(\mathbf{a},~\mathbf{p},~\mathbf{q})$ to $E^T(\mathbf{a},~\mathbf{p},~\mathbf{q})$ ratio. To find this specific value of $\epsilon$, we use the equality below:
	\begin{equation}\label{eps}
	\epsilon = \delta R_\mathrm{max},
	\end{equation}
	where $\delta$ is a positive value in the range of ($0,1$]. Depending on the value of $\delta$, the ratio between system's throughput and energy consumption varies, however, for an specific $\delta$ this ratio reaches a maximum value. This observation is due to the fact that EE is by itself a trade-off between system's throughput and energy consumption. Therefore, by testing different values of $\delta$ (different $E^T(\mathbf{a},~\mathbf{p},~\mathbf{q})$ to $\tilde{R}^T(\mathbf{a},~\mathbf{p},~\mathbf{q})$ ratio), we can find the point in which maximum EE is achieved.

\subsection{Computational Complexity Analysis}
	In this subsection, we investigate the computational complexity of our proposed algorithm.~In optimization problem~(\ref{last}),~we have $3KNB$ decision variables and $B(1+K)+N(K+B+3KB+2)$ convex constraints.~Hence,~the time complexity of this problem is of order $\mathcal{O}(KNB)^{3}(B(1+K)+N(K+B+3KB+2))$~which is polynomial.~It is worth mentioning that an exhaustive search for an optimal
	scheme, merely for subchannel allocation,~would require the examination of all $BN^{K}$ possible choices,~which is considerably higher than the complexity of our proposed solution. Furthermore,~by employing D.C. programming with the interior point method, the approximate number of required iterations would be $$I=\frac{\log(\frac{B(1+K)3+N(K+B+3KB+2}{{t}^{0}\varrho})}{\log(\mu)},$$~where $t_{0}$ is the initial point, $0\leq\varrho\ll 1$ is the stopping criterion,~and $\mu$ is used for updating the accuracy of the method \cite{WCL,Ata_TWC}.  
	%\vspace{-13mm}
	
	\section{Simulation Results and Discussions}
	In this section, we evaluate the performance of our proposed resource allocation algorithm through extensive simulations and compare the performance with the existing research works in [13], [14] , [15]. In our simulations, we consider OFDMA small-cell network with $B$ = 3 SBSs each with radius $100$ m and $K=32$ subchannels. We further assume that there are $N=12$ users randomly distributed in the network. The channel gain between a transmitter and a receiver is calculated using independent and identically distributed Rayleigh flat fading and the figures shown in this section are obtained by estimating the average of results over different realizations of path-loss as well as multi-path fading. \textcolor{black}{The large scale fading of the communication channel is computed according to the path loss formula and the small scale fading is modeled by Rayleigh fading and is formulated as Path-loss=$P_{L_{0}}+10 \theta \log (d)$,~where $d$ denotes the distance between user and BS $P_{L_{0}}$is the constant path-loss coefficient (128.1) dB which depends on the antenna characteristics, and $\theta$ denotes the path-loss exponent(in our case $\theta$ is equal to 3}. Without loss of generality, we assume that SBSs' and users' SI-cancellation factors are the same and $\Delta_\mathrm{u}$ = $\Delta_\mathrm{bs}$ = $\Delta$ = -70 dB and the step size $\nu$ in Algorithm. 1 is set to 0.1. The remaining parameters are given in Table I.
		\begin{center}
		\begin{table}
			\centering
			\caption{Simulation Parameters}
			\renewcommand{\arraystretch}{1.35}
			\begin{tabular}{|>{\centering\arraybackslash}p{3.8cm} |>{\centering\arraybackslash}p{2.8cm}|}
				\hline
				Parameters & Value \\ \hline\hline
				$\sigma^2$ & -120 dBm\\  \hline
				$p_{\max}$ & 32 dBm\\  \hline
				$p_{\max}^u$ & 23 dBm\\ \hline
				$R_\mathrm{min}^\mathrm{dl}$  & 4 bps/Hz\\ \hline
				$R_\mathrm{min}^\mathrm{ul}$  & 2 bps/Hz\\ \hline
				$P_\mathrm{c}^\mathrm{u}$  & 0.1 W\\ \hline
				$P_\mathrm{c}^\mathrm{SBS}$  & 1 W\\ \hline
				$\kappa$  & 38\%\\ \hline
				$\psi$  & 20\% \\ \hline
				$\lambda$  & $10^6$ \\ \hline
				Path loss exponent  &3 \\ \hline
			\end{tabular}
		\end{table}
	\end{center}
	\subsection{Impact of SI Cancellation}
	We first examine the effect of SI-cancellation factor, $\Delta$, on energy efficiency of IBFD networks. In Fig. 2, system energy efficiency vs. $\delta$  for different values of $\Delta$ is presented. We also draw a comparison between EE of IBFD and HD communications. For HD case, we assume that half of the existing subchannels are reserved for DL and the other half for UL communications, exclusively. In HD case, after dealing with the non-convex feasible set of EE maximization problem according to our proposed method, the resulting problem is dealt using Dinkelbach method which attains the optimal solution.
	
	As observed in Fig. 2, by decreasing $\Delta$, system EE would increase. This is due to the fact that lower values of $\Delta$ correspond to lower SI and thus higher EE. Furthermore, in each IBFD case, for a specific $\delta$, EE reaches its peak and then decreases. However, the value of $\delta$ for which maximum EE is obtained, varies from one case to another. For instance, when $\Delta = -110$ dB,  the maximum EE is achieved when $\delta = 0.6$ while for $\Delta = -70$ dB, system EE peaks at $\delta = 0.4$. This observation can be explained by considering the amount of data rate that a user can attain by consuming a unit of energy. When $\Delta = -110$ dB, because of the lower SI, users would be able to achieve a notable data rate, even while transmitting with a nominal transmit power. In this case, since the substantial growth in system throughput is worth the slight increase in system power consumption, the $\delta$ for which maximum EE is attained leans toward higher values. 
	
	In contrast, when interference is high, the value of $\delta$ corresponding to the maximum EE would get closer to its lower values. Another important observation in Fig. 2  is the superiority of IBFD communications'  performance compared to HD in most cases. This improved performance is the result of the higher flexibility of spectrum usage in IBFD communications. Note that, as $\delta$ gets closer to its optimal value (in peaks), the EE achieved using IBFD becomes higher than EE of HD in all but one case which is when $\Delta= -50$ dB. In $\Delta= -50$ dB, SI cancellation is too low and thus the performance of HD outperforms IBFD communications. 
	\begin{figure}
		\centering
		\includegraphics[width=9.00cm,height=6.00cm]{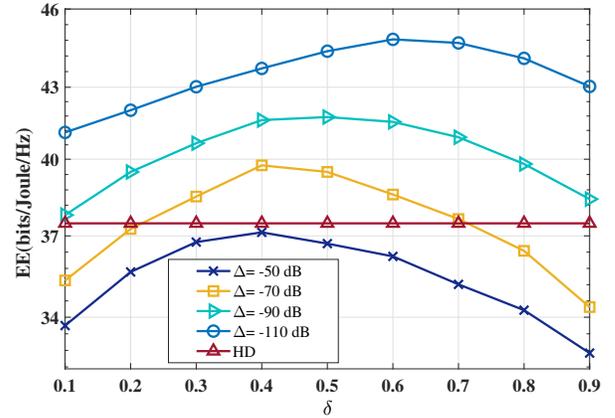}
		\caption{System energy efficiency vs. $\delta$ for different values of SI cancellation factor $\Delta$.}
	\end{figure}
	\begin{figure}
		\centering
		\includegraphics[width=9.00cm,height=6.00cm]{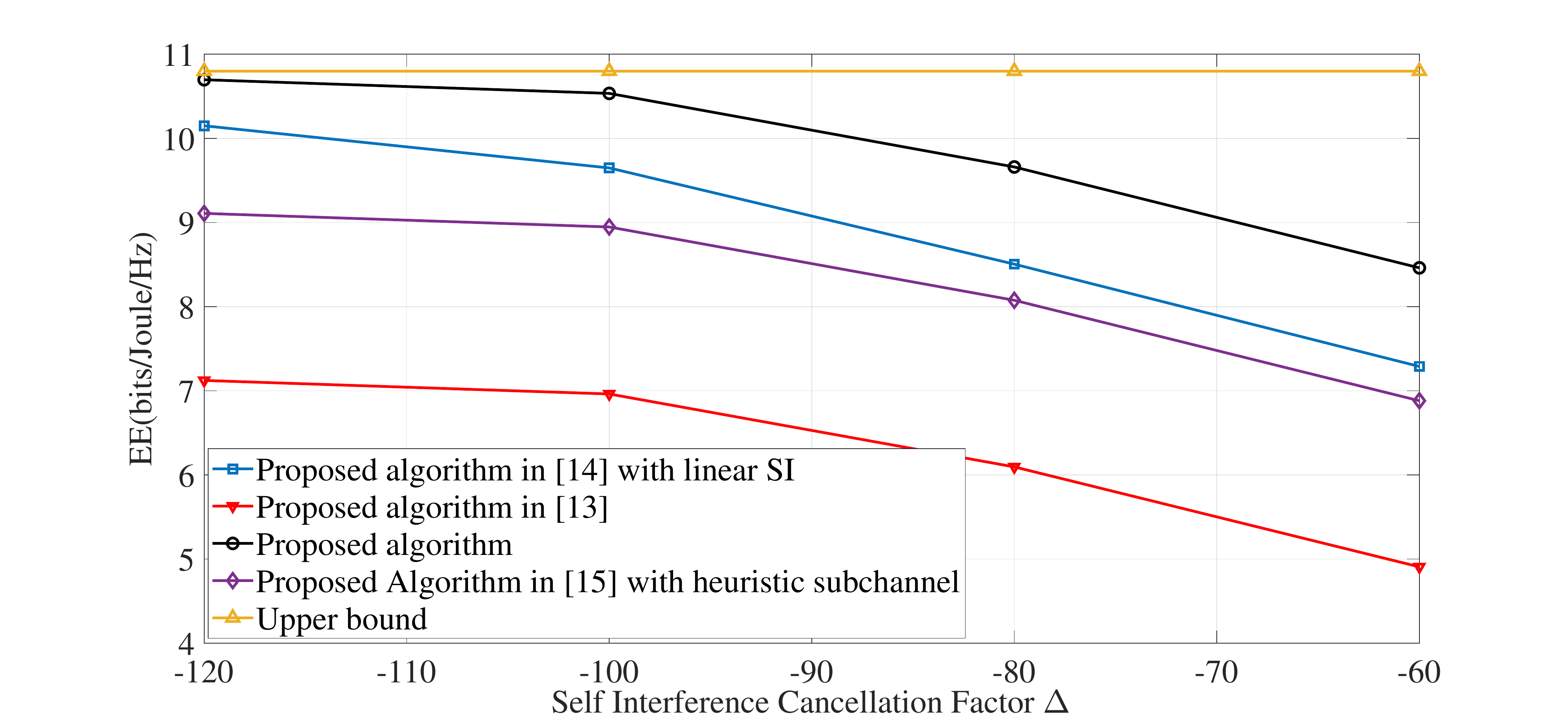}
		\caption{System energy efficiency vs. SI cancellation factor  $\Delta$ compared to the existing benchmark schemes and upper bound.}
	\end{figure}
	
	\subsection{Comparison with the Existing Benchmarks}
	Due to the substantial impact of SI-cancellation factor, $\Delta$, on the EE of IBFD networks, in Fig.~3, we compare the performance of our algorithm (in single-cell scenario) with three other existing schemes and an upper bound, while considering different values for $\Delta$. To derive an upper bound for system EE and draw comparisons between our proposed algorithm and some of the prominent relevant works, in Fig.~3 we consider a single cell network instead of the multi-cell network and ignore SI to obtain the solution of (\ref{last-1}) under this new assumption. 
	We also simulate the proposed algorithm in [15], where a trade-off between EE and SE in a cell with an IBFD BS, is derived. We should also note that even though in [13], SI is presented as an additive factor of the background noise, to make a fair comparison, here we consider that SI is modeled as an linear function of the transmit power (similar to the model used both in [15] and this paper). Since, in our proposed algorithm, we eventually restated our MOOP as an equivalent single objective problem aiming at minimizing system energy consumption, here we also compare the result of our algorithm with that of [14], in which minimization of system aggregate power is investigated.

	As can be seen in Fig. 3, for lower values of $\Delta$, the performance gap between EE of the upper bound and that of our algorithm is quite nominal. However, for higher values of $\Delta$, the increasing SI results in performance degradation in our proposed algorithm, while the upper bound case remains immune to $\Delta$. Thus, the performance gap gets larger as the value of $\Delta$ increases. Furthermore, we can perceive from Fig. 3 that our proposed algorithm can noticeably improve the network EE, as compared to both [12] and [15]. The considerable EE improvement  compared to [12] is due to two factors. Firstly, in [12] a heuristic approach is proposed for subchannel allocation which does not necessarily result in the optimal solution. Secondly, the Augmented Lagrangian method which is used for power allocation, does not work well when interference is taken into account (non-convex optimization problem) and converges to a sub-optimal solution. As for the superiority of our solution over that of [15], we should underline the fact that in [15], subchannel and power allocation problem is decomposed and subchannel allocation is again obtained through a heuristic algorithm. Needless to say, this decomposition can considerably decrease the achieved system EE. 
	
	\begin{figure}
		\centering
		\vspace{-3mm}
		\includegraphics[width=9.00cm,height=6.00cm]{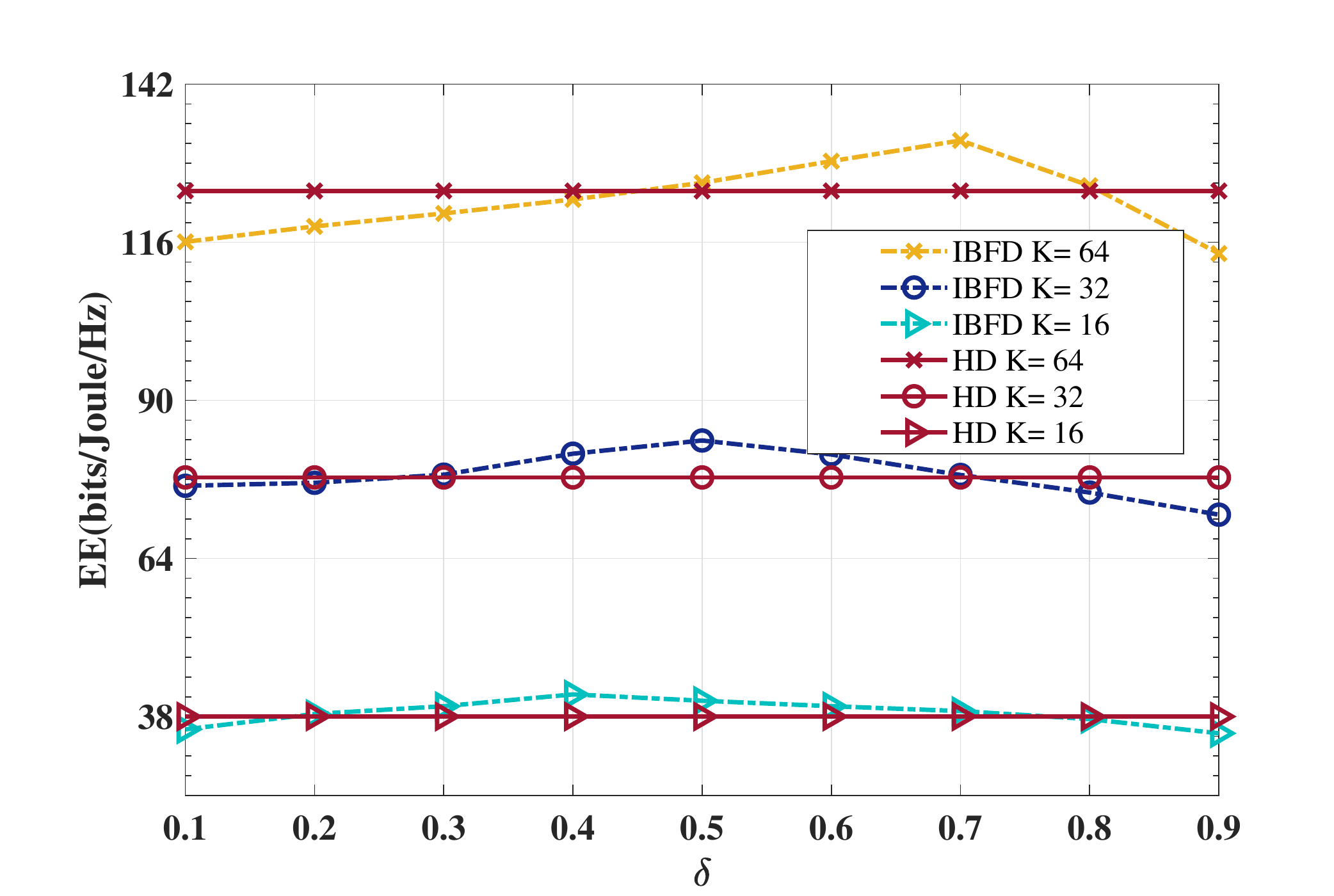}
		\caption{System energy efficiency vs. $\delta$ for different values of the number of subchannels $K$.}
	\end{figure}
	
	\subsection{Impact of the Number of Subchannels}
	In Fig. 4, the impact of number of subchannels on system EE is investigated. 
	%For this figure, since when $K = $ 16, in HD only 8 subchannels would be available for the communications in either directions (UL or DL), we consider that the number of users, $N$, is 8.
	It is obvious from this figure that as the number of subchannels increases, system EE improves as well. This is due to the fact that more subchannels in the network means that there are more available subchannels that can be exploited for improving system EE. Moreover, similar to what was previously explained for Fig. 2, around the optimal value of $\delta$ the performance of IBFD exceeds that of HD. However, the values of $\delta$ for which the EE of IBFD surpasses HD differs from one case to another. For instance, compare $K = 64$ and $K = 16$. When $K = 16$, from $\delta = $ 0.2 up to $\delta= $ 0.8, the EE of IBFD becomes higher than HD, while for $K = $ 64, only for $\delta\in$ \{0.5, 0.6, 0.7\}, IBFD EE exceeds that of HD. This observation once again underlines the significance of flexible spectrum usage that is achieved in IBFD communications.  
	
	Consider the HD case when $K = $ 16, where we have  only 8 subchannels for UL and DL, which subsequently restricts users' access to a channel with a desirable condition. On the other hand, in IBFD communications, it is possible for some users to obtain their minimum data rate requirements in both directions, using only a single subchannel. Therefore, some of the subchannels may remain unallocated which can be assigned to those users that have good channel condition in them and can improve the overall network EE.

	\begin{figure}
		\centering
		\includegraphics[width=9.00cm,height=6.00cm]{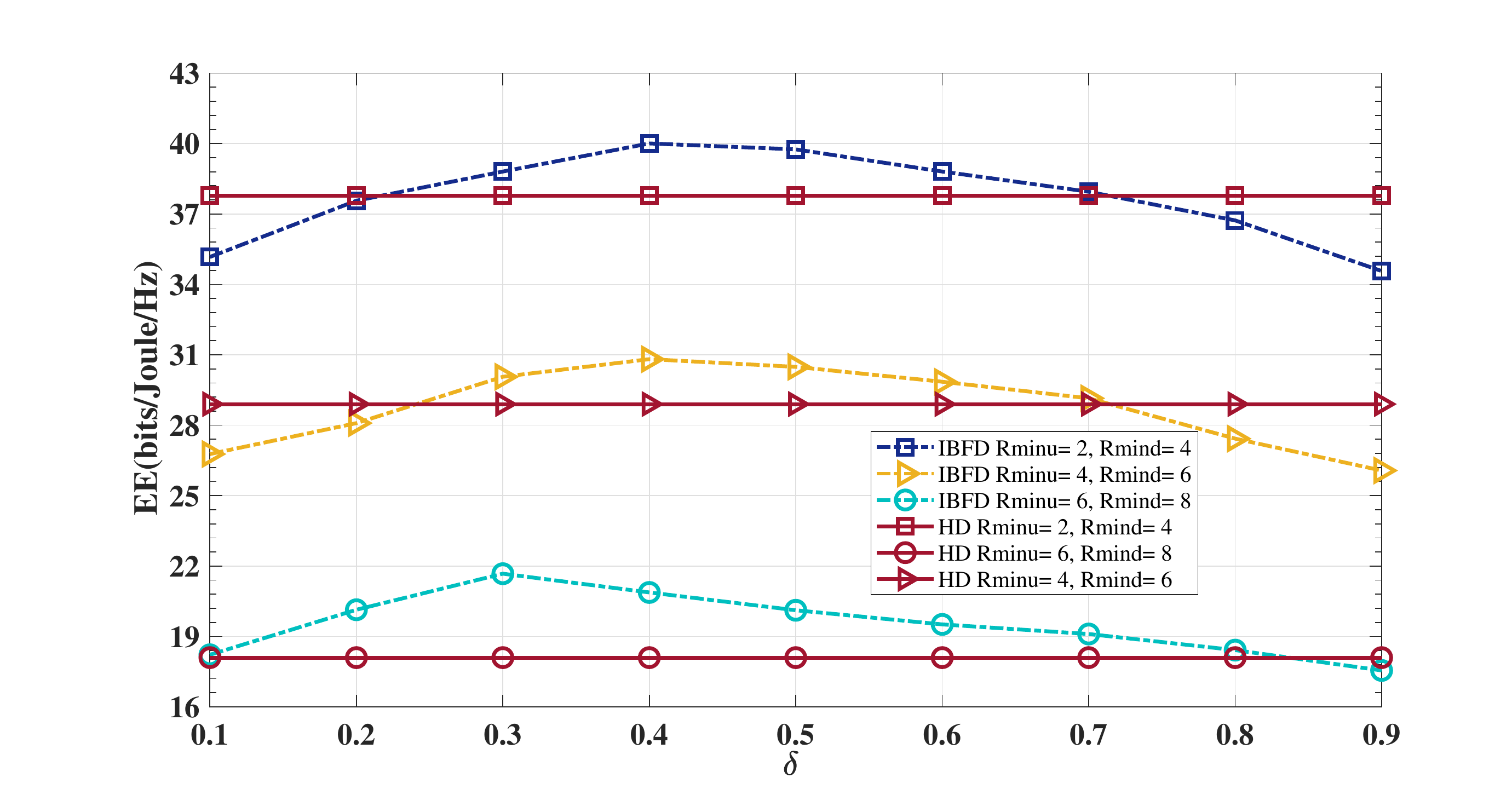}
		\caption{System energy efficiency vs. $\delta$ for different QoS requirement.}
	\end{figure}
	\subsection{Impact of QoS requirements}
	In Fig. 5, network EE for various QoS requirements is depicted for both IBFD and HD communications. In Fig. 5, it is illustrated that as users' QoS requirement increases, EE decreases. The reason is that, when QoS requirement is high, more subchannels have to be allocated to users, especially those with poor overall channel condition, to meet their minimum rate requirements. This, in conjunction with necessity of higher transmit power for reaching a specific data rate, would increase system energy consumption and decrease the achievable data rate of system (lower achievable EE). Furthermore, comparing the first two cases when minimum DL data rate requirements are 4 and 6, respectively, with their corresponding HD case, we can see that the performance gain from IBFD is much higher in the first case. In fact as mentioned before, since in IBFD communications we have to deal with additional interference (such as user to user and BS to BS interference), when level of interference gets higher due to the increased QoS requirements, the performance degradation in IBFD communication would be more intense. However this explanation is not applicable when we compare the trend of EE  when $R_\mathrm{min}^\mathrm{dl}$ = 4 with EE of the case where $R_\mathrm{min}^\mathrm{dl}$= 8 (in which the EE of IBFD communications surpasses the EE of HD for almost all values of $\delta$). In fact, when $R_\mathrm{min}^\mathrm{dl}$= 8, the minimum data rate requirement is high enough to result in scarcity of available resources. Thus, as can be seen, in this case, the EE of IBFD is almost always higher than that of HD.
	\begin{figure}
		\centering
		\includegraphics[width=9.00cm,height=6.00cm]{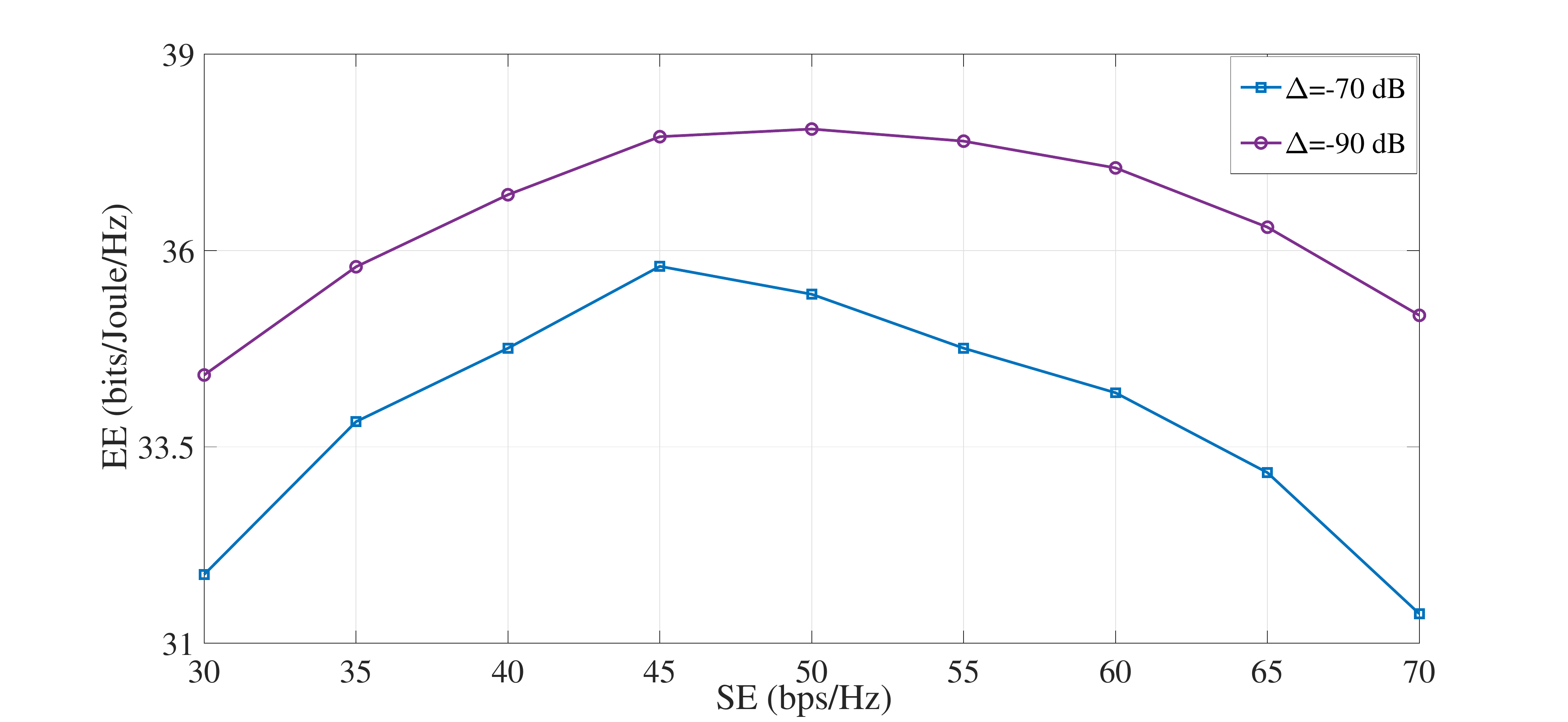}
		\caption{System energy efficiency vs. spectral efficiency for SI cancellation factors.}
	\end{figure}
		\subsection{Performance Tradeoff Between EE-SE}
	Fig. 6 illustrates the trend of system EE with respect to SE. It can be perceived that as the throughput of network increases, EE starts to increase first and then sharply decreases. In fact, system throughput is by itself a function of system transmit power, thus any increase in the throughput may results in higher energy consumption as well. The point at which energy  consumption exceeds SE gains, the overall EE starts reducing.~We also observe that the value of trade-off region decreases, as the SI cancellation factors decrease.~The reason is straightforward and it is because that the low level of SI cancellation factors spend more power which results in a degradation in the system SE as well as system EE. 
	\begin{figure}
		\centering
		\vspace{-3mm}
		\includegraphics[width=9.00cm,height=6.00cm]{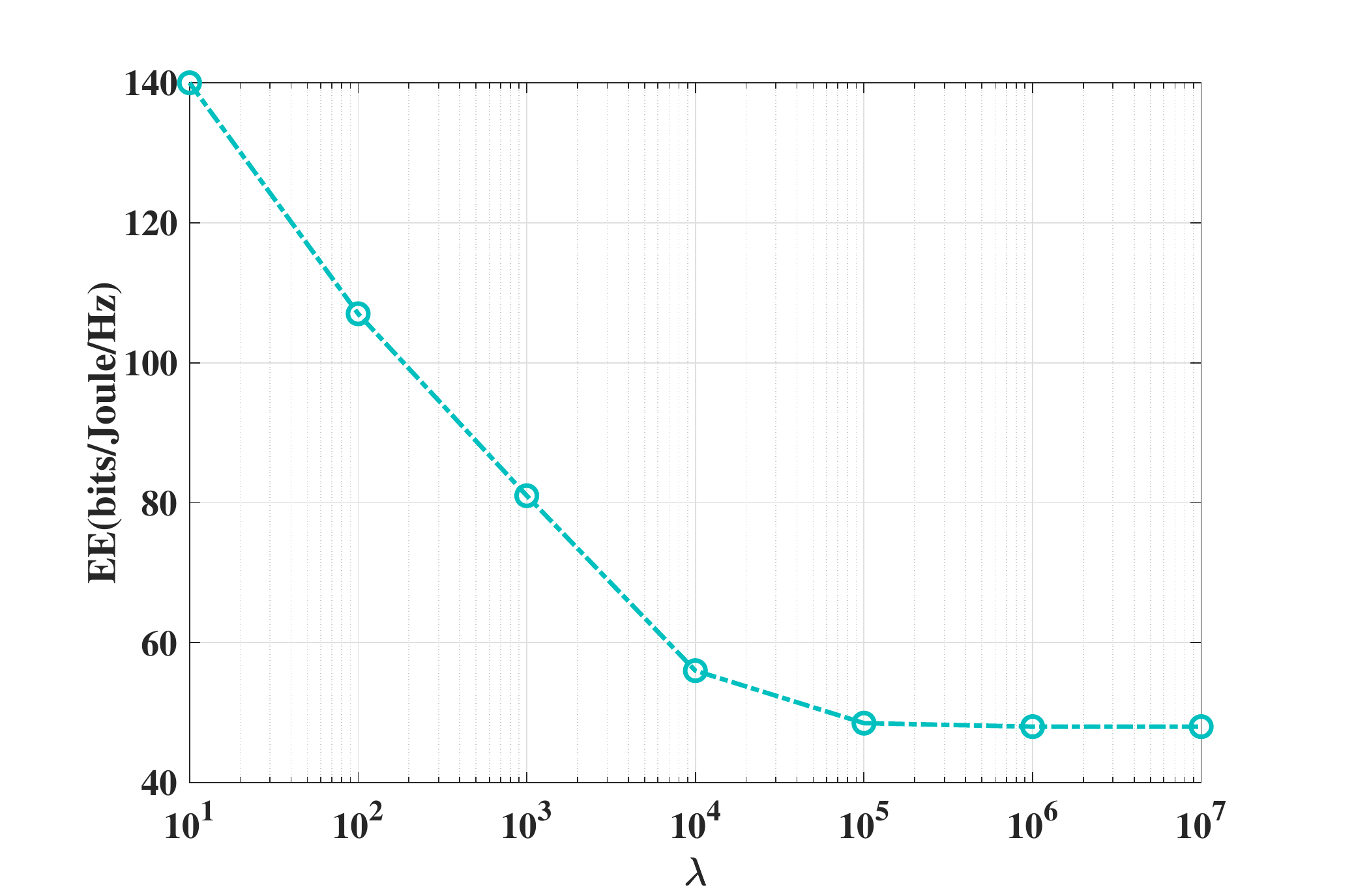}
		\caption{System energy efficiency vs. $\lambda$}
		\vspace{-5mm}
	\end{figure}
	
	\subsection{Impact of the Weight of Penalty Function}
	Another important parameter in our simulations is the parameter $\lambda$. This parameter works as the weight of the penalty function in the objective function of (\ref{last}), in order to ensure that the values of subchannel allocation variables would belong to the set \{0, 1\}. In Fig. 7, the impact of this parameter on system EE is illustrated. At first, since $\lambda$ is small and thus the cost of having non-binary subchannel allocation variables is small, the value of these variables would not comply with their binary nature. This means that at least for some of the subchannels, the equality $a_{n,b,k}- a^2_{n,b,k} = 0 $ would not hold. In this case, it is possible for a portion of subchannel $k$ to be assigned to one user while the other portion is assigned to others. This would increase the achievable data rate of users and decreases system energy consumption, however, the OFDMA nature of network would be violated. As $\lambda$ increases, each subchannel would be allocated to at most one user and the value of subchannel variables would get closer to 0 or 1. This would subsequently cause the reduction of system EE, as the number of subchannels that a transceiver can use for sending its data on, would be restricted. Nevertheless, after $\lambda$ reaches a specific value, here $10^6$, EE converges to its final value and remains unchanged regardless of any further increase in $\lambda$. This is due to the fact that when $\lambda$ is sufficiently large, the penalty term would converge to zero and from there on, increasing its weight would become negligible. 
	
	\subsection{Convergence of Algorithm~1}
	The convergence of our proposed algorithm is investigated in three cases in Fig.~8. In the first case, we assume that users' transmit power is initialized by their maximum transmit power divided over number of subchannels, whereas in the second and third case the initial transmit powers are set to the maximum transmit power and zero, respectively. Even though the rate of convergence differs from one case to another, in all the three scenarios only after a limited number of iterations our algorithm converges to a locally optimal value.

	\begin{figure}
		\centering
		\vspace{-3mm}
		\includegraphics[width=10.00cm,height=6.00cm]{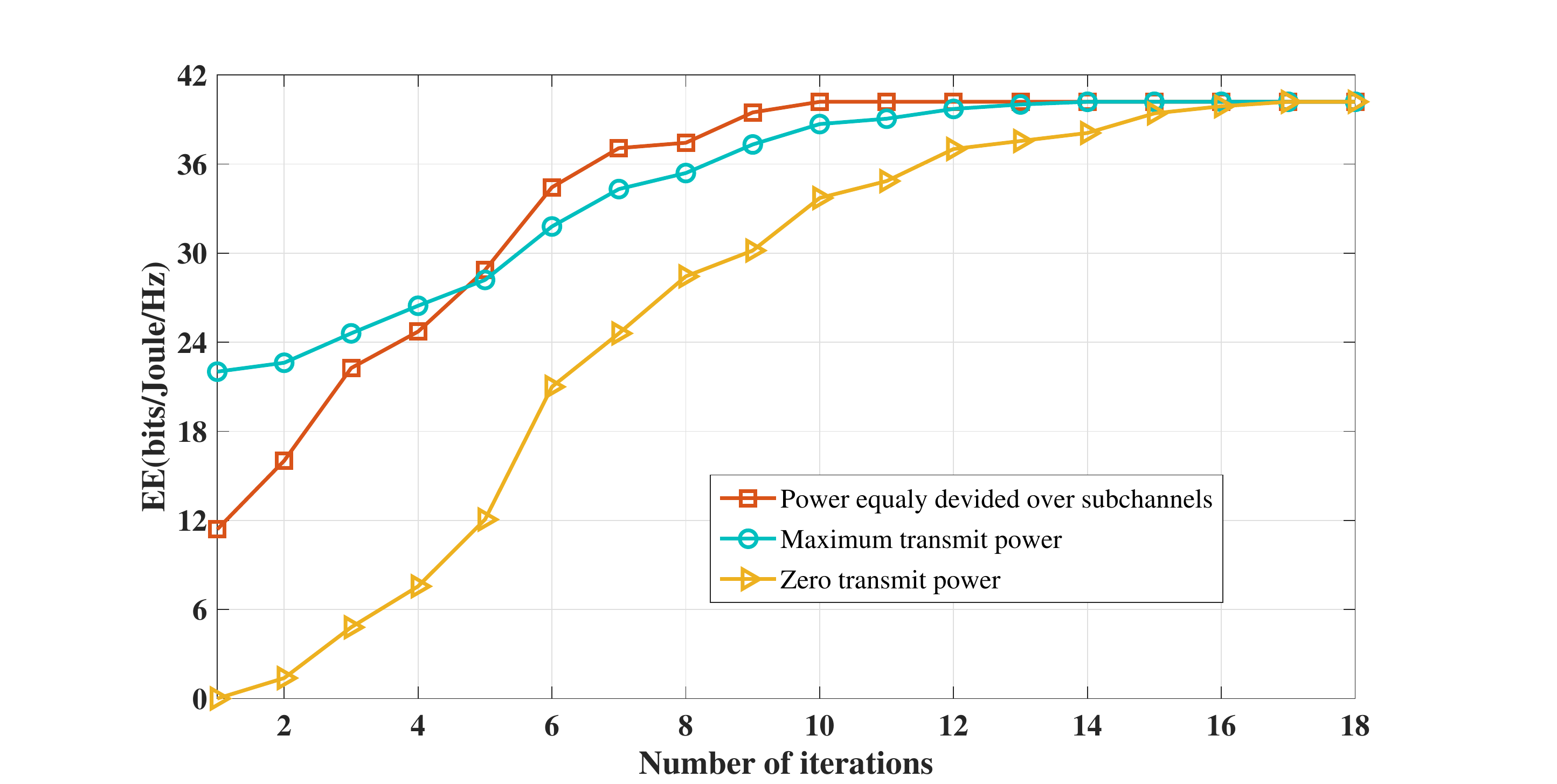}
		\caption{The convergence rate of Algorithm 1}
	\end{figure}
	\section{Conclusion}
	In this paper, we investigated the problem of joint BS,~subchanel assignment and power control for maximization of EE in an OFDMA small-cell network with IBFD communications. This problem was then reformulated into a MOOP which enables us to prioritize system throughput maximization and power consumption minimization, depending on our applications' prerequisites. To obtain all the Pareto fronts in the aforementioned problem, we used $\epsilon$-constraint method. Furthermore, in order to tackle the non-convexity of the constraint set, MM approach was employed for approximating the non-convex rate functions and a penalty function was also introduced to handle the binary joint BS,~subchannel assignment variables. In the simulation results, the effectiveness of IBFD communications, as well as the capability of our proposed solution in improving EE of network were demonstrated through extensive simulations.

\begin{appendices}
\section{Proof of Proposition 1}
\textcolor{black}{In order to give a formal proof of Proposition 1, a
general fractional programming which can be formulated as below is considered:}
\begin{equation}
\textcolor{black}{
\min_{x} \text{ \ }\theta (x)=\frac{f(x)}{g(x)}:x\in X  \label{th1}}
\end{equation}
\textcolor{black}{where $X$ is a nonempty compact set belonging to $R^{n}$. Furthermore, $f(x)$ and $g(x)$ are continuous
real-valued functions of $x\in X$ and $g(x)>0$, for all $x\in X.$ To address the optimal solution, we define the following function:
\begin{equation}
\textcolor{black}{H(\psi^{\ast })=\min_{x}\left\{ f(x)-\psi^{\ast }g(x):x\in X\right\}}   \label{th4}
\end{equation}%
as the minimum value of $f(x)-\psi g(x)$ with each fixed $\psi^{\ast }.$\\
Based on Dinklebach method \cite{R}, it is proved that
\begin{equation}
\psi^{\ast }=\frac{f(x^{\ast })}{g(x^{\ast })}=\min_{x}\left\{ \frac{f(x)}{g(x)}%
:x\in X\right\}   \label{th2}
\end{equation}%
if and only if
\begin{equation}
H(\psi^{\ast })=H(\psi^{\ast },x^{\ast })=\min_{x}\left\{ f(x)-\psi^{\ast }g(x):x\in
X\right\} =0.  \label{th3}
\end{equation}
Therefore, from (\ref{th2}) and (\ref{th3}), it can be concluded that the
optimal solution $x^*$ of (\ref{th1}) is the optimal solution of (\ref{th4})\
when $\psi=\psi^{\ast }$, where $\psi^{\ast }$ denotes the minimum value of (\ref{th1}%
).\\
On the other hand, we now formulate a general MOOP with two competing objectives as
follows:
\begin{eqnarray}
& \min \text{ \ \ }f(x)  \label{mup1} \\
& \max \text{ \ \ }g(x)  \label{mup2} \\
& \mathrm{s.t.}\text{ \ \ \ }x>0  \label{mup3}\nonumber
\end{eqnarray}
It should be noted that $f(x)$ is the numerator of fractional optimization problem in \ref{th1} while $g(x)$ is its denominator. By combining the competing objective functions (\ref{mup1}) and (\ref%
{mup2}) into a single objective function  through $\epsilon$-method, the objective functions in the MOOP can be changed into a
single optimization as:
\begin{align}\label{mup4}
&~\min f(x)\\ 
&~g(x)\geq \epsilon\\\nonumber
&\mathrm{s.t.}~~x>0 \nonumber
\end{align}%
By comparing \eqref{mup4} and \eqref{th2}, one can easily verify that optimal set of %
\eqref{mup4} is inclusive of the solution for \eqref{th2}. The value of $\epsilon
$ that makes the minimum of the introduced single objective optimization, would clearly yield a solution
for the fractional programming problem as well.
}
\section{Proof of Proposition 2}
	We start this proof by first demonstrating that the solution obtained for optimization problem (\ref{last}), through our proposed algorithm, is a tight upper-bound for the problem (\ref{last+1}).
	
	First we should point out that as demonstrated in \cite{lenog}, the rate approximations given in (31) and (32), are tight lower bounds of the original rate functions. In conjunction with this fact, we also have to consider that in the $t^\mathrm{th}$ iteration, the objective function of (\ref{last}) would be $e_1(\mathbf{a}^t,~\mathbf{p}^t,~\mathbf{q}^t) - \tilde{e}_2(\mathbf{a}^t)$. Therefore, we have:
	\begin{eqnarray}
		\begin{aligned}
			e_1(&\mathbf{a}^{t+1},~\mathbf{p}^{t+1},~\mathbf{q}^{t+1})- {e}_2(\mathbf{a}^{t+1})\leq e_1(\mathbf{a}^{t+1},~\mathbf{p}^{t+1},~\mathbf{q}^{t+1})\\&-{e}_2(\mathbf{a}^t) -\nabla_{\mathbf{a}}e_{2}^{T}(\mathbf{a}^{t}).(\mathbf{a}^{t+1}-{\mathbf{a}}^{t})
			=\min_{\mathbf{a},\mathbf{p},\mathbf{q}}e_1(\mathbf{a},~\mathbf{p},~\mathbf{q})\\&-{e}_2(\mathbf{a}^t) -\nabla_{\mathbf{a}}e_{2}^{T}(\mathbf{a}^{t}).(\mathbf{a}-{\mathbf{a}}^{t})\leq e_1(\mathbf{a}^{t},~\mathbf{p}^{t},~\mathbf{q}^{t})\\&-{e}_2(\mathbf{a}^t)-\nabla_{\mathbf{a}}e_{2}^{T}(\mathbf{a}^{t}).(\mathbf{a}^{t}-{\mathbf{a}}^{t})\\
			&=e_1(\mathbf{a}^{t},~\mathbf{p}^{t},~\mathbf{q}^{t})-{e}_2(\mathbf{a}^t).
		\end{aligned}	
	\end{eqnarray}
	
	One can verify that as the DC iterations continue, the objective function of (\ref{last}) takes smaller values (until convergence).~Eventually, when $\mathbf{a}^t=\mathbf{a}^{t-1}$, $\mathbf{p}^t=\mathbf{p}^{t-1}$, and $\mathbf{q}^t=\mathbf{q}^{t-1}$,~the non-equality above would hold for equality. From the previous two facts we can conclude that the obtained solution for (\ref{last}) would be a tight upper bound for (\ref{last+1}) as well.
	
	Furthermore, from \textbf{Remark 2} we know that the optimization problem (\ref{last+1}) is equivalent to (\ref{third}). Thus, the solution obtained for (\ref{last+1}) would also be a locally optimal point for (\ref{third}). 
		Finally, from Remark~1, we conclude that for the right value of $\epsilon$, the solution attained for (\ref{third}) would also be a locally optimal point for (\ref{second}), which is by itself equivalent to the problem (\ref{first}).

	\end{appendices}
	\bibliographystyle{IEEEtran}
	\bibliography{IEEEabrv,ref}
	%\bibliography{ref}
		\begin{IEEEbiography}[{\includegraphics[width=1.1in,height=1.45in]
	{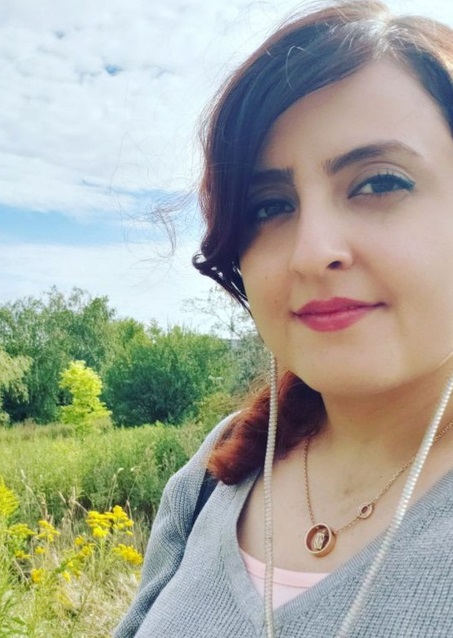}}]
	{Sheyda Zarandi} received her M.Sc. degree
in Information Technology with focus on Computer Networks from the Amirkabir University
of Technology, Tehran, Iran, in
2014 and 2017.
	From 2019 she has continued her study and research at York university, Toronto, Canada. Her research
	interests are mainly focused on resource allocation in wireless communication, Edge learning methods,
	green communication, mobile edge computing.
\end{IEEEbiography}
	\begin{IEEEbiography}[{\includegraphics[width=1.1in,height=1.35in]
	{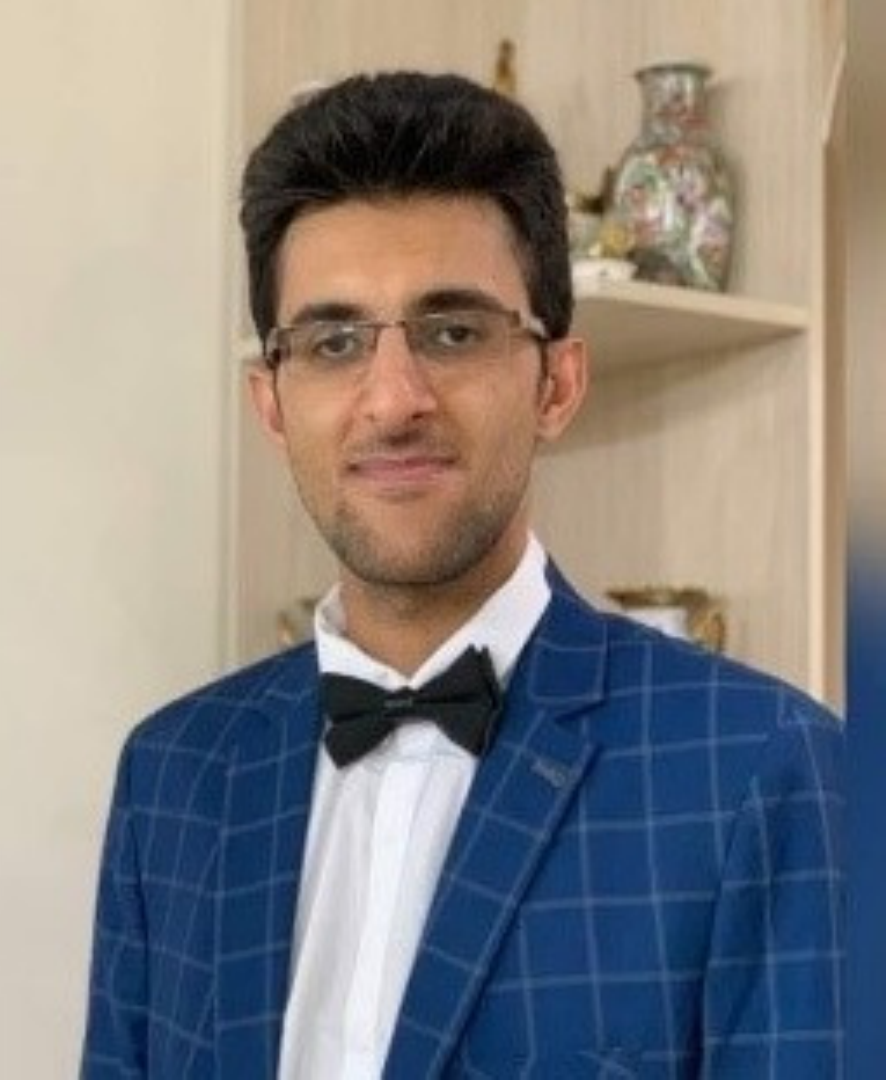}}]
	{Ata Khalili} (S'18) received the B.Sc. degree and M.Sc. degree with first class honors in Electronic
	Engineering and Telecommunication Engineering from Shahed University in 2016 and 2018, respectively.
	From 2018 until 2019, he was a visiting researcher at the Department of Computer Engineering and Information Technology, Amirkabir University of Technology, Tehran, Iran. Since Oct 2019, he has been at the Department of Electrical and Computer Engineering, Tarbiat Modares University, Tehran, Iran, where he is currently a research assistant. He is
	also working as a research assistant at the Electronics
	Research Institute, Sharif University of Technology, Tehran, Iran. His research
	interests include intelligent reflecting surface (IRS), unmanned aerial vehicle (UAV) communications, resource allocation in wireless communication,
	green communication, mobile edge computing, and optimization theory.~He served as a member of Technical Program Committees for the IEEE Globecom Conference in 2020.
\end{IEEEbiography}

\begin{IEEEbiography}[{\includegraphics[width=1.1in,height=1.in,clip,keepaspectratio]{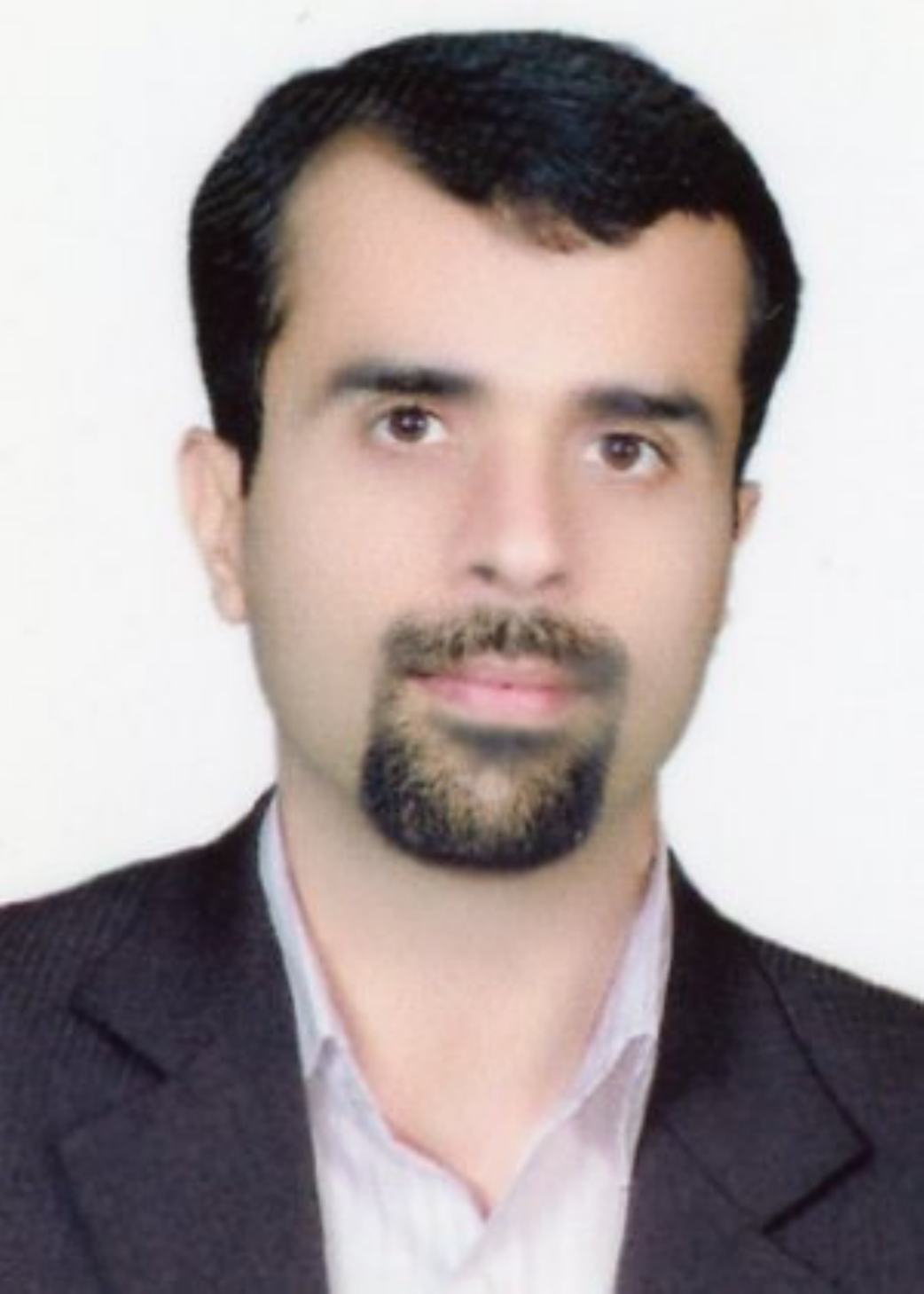}}]
{Mehdi Rasti} (S'08-M'11) is currently an Assistant Professor at the Department of Computer Engineering, Amirkabir University of Technology, Tehran, Iran. From November 2007 to November 2008, he was a visiting researcher at the Wireless@KTH, Royal Institute of Technology, Stockholm, Sweden. From September 2010 to July 2012 he was with Shiraz University of Technology, Shiraz, Iran. From June 2013 to August 2013, and from July 2014 to August 2014 he was a visiting researcher in the Department of Electrical and Computer Engineering, University of Manitoba, Winnipeg, MB, Canada. He received his B.Sc. degree from Shiraz University, Shiraz, Iran, and the M.Sc. and Ph.D. degrees both from Tarbiat Modares University, Tehran, Iran, all in Electrical Engineering in 2001, 2003 and 2009, respectively. His current research interests include radio resource allocation in IoT, Beyond 5G and 6G wireless networks.
\end{IEEEbiography}

\begin{IEEEbiography}[{\includegraphics[width=1.1in,height=1.35in]{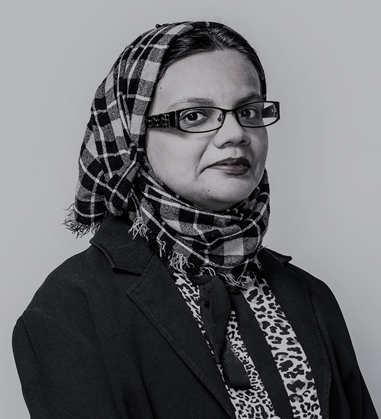}}]
{Hina Tabassum} (SM'17) Hina Tabassum is currently an Assistant Professor at the Lassonde School of Engineering, York University, Canada. Prior to that, she was a postdoctoral research associate at the Department of Electrical and Computer Engineering, University of Manitoba, Canada. She received her PhD degree from King Abdullah University of Science and Technology (KAUST). She is a Senior member of IEEE and registered Professional Engineer in the province of Ontario, Canada. She has been recognized as an Exemplary Reviewer (Top $2\%$ of all reviewers) by IEEE Transactions on Communications in 2015, 2016, 2017, and 2019. Currently, she is serving as an Associate Editor in IEEE Communications Letters and IEEE Open Journal of Communications Society. Her research interests include stochastic modeling and optimization of wireless networks including vehicular, aerial, and  satellite networks, millimeter and terahertz communication networks, software-defined networking and virtualized resource allocation in wireless networks.
\end{IEEEbiography}

\end{document}